\documentclass [11pt]{article}
\usepackage{amsfonts}
\usepackage{amsmath,amssymb,amsthm}
\usepackage{graphicx}
\usepackage{bbm}
\usepackage[all]{xy}
\usepackage{rotating}
\usepackage{multirow}
\usepackage{enumerate}
\usepackage[dvipsnames]{xcolor}
\usepackage{tikz}
\usetikzlibrary{arrows,intersections,arrows.meta,positioning,fadings,calc}
\usepackage{todonotes}
\usepackage[a4paper,textwidth=16cm,textheight=24cm]{geometry}
\usepackage[authoryear,round]{natbib}
\usepackage{authblk}

\usepackage[colorlinks,linkcolor=RoyalBlue,citecolor=RoyalBlue,urlcolor=RoyalBlue]{hyperref}

\newtheorem{theorem}{Theorem}
\newtheorem{remark}[theorem]{Remark}

\newtheorem{proposition}[theorem]{Proposition}
\newtheorem{corollary}[theorem]{Corollary}
\newtheorem{definition}[theorem]{Definition}
\newtheorem{example}[theorem]{Example}

\newcommand{\R}{\mathbb{R}}
\newcommand{\F}{\mathbb{F}}
\newcommand{\xbar}[1]{\overline{X}_{#1}}
\newcommand{\Fbar}{\overline{F}}

\newcommand{\BB}{\mathbb{B}_F}
\newcommand{\conv}[2]{#1^{\ast #2}}
\newcommand{\iid}{{\rm i.i.d.}}
\newcommand{\thetabar}{\overline{\theta}}
\newcommand{\etabar}{\overline{\eta}}

\newcommand{\InvSub}{{\rm InvSub}}

\begin{document}

\title{Nonparametric tests for stochastic dominance between linear combinations of risks}

\author[1]{Tommaso Lando\thanks{tommaso.lando@unibg.it, ORCID: 0000-0003-4288-0264}}
\author[2]{Paulo Eduardo Oliveira\thanks{paulo@mat.uc.pt, ORCID: 0000-0001-727-5705}}
\affil[1]{Department of Economics, University of Bergamo, Italy}

\affil[2]{CMUC, Department of Mathematics, University of Coimbra, Portugal}

\date{}

\maketitle

\begin{abstract}
Weighted aggregations of \iid\ risks without a finite mean can behave in a strikingly different way from the finite-mean case: as the weight vector becomes more balanced, the resulting combination may become stochastically larger, rather than less dispersed. Existing results establish stochastic dominance between pairs of linear combinations—or between a convex combination and the underlying variable—under shape restrictions on the distribution and structural assumptions on the weights. Nonetheless, two practical limitations remain: (i) the sufficient conditions vary across results, and (ii) being non-necessary, they exclude many relevant configurations. Moreover, under a statistical perspective, where the true distribution of the data is assumed to be unknown, these conditions cannot be checked. Motivated by this gap, we develop nonparametric procedures to test whether two linear combinations are stochastically ordered. We propose two complementary approaches: a least-favorable calibration and a bootstrap-based method. We establish their asymptotic validity under the null of stochastic dominance and their consistency against the alternative of non-dominance. Monte Carlo experiments illustrate the finite-sample performance of the proposed procedures across a range of models and weight configurations. An application to insurance claim-severity data illustrates how the tests can be used to assess whether pooling leads to a stochastically larger loss.
\end{abstract}




\section{Introduction}\label{intro}
Risk aggregation is a central problem in insurance, finance, and risk theory. 
When several losses or risk exposures are combined through a vector of weights, the 
usual intuition is that spreading the exposure more evenly should reduce risk. This 
diversification principle is consistent with classical results in stochastic ordering 
theory. It is well-known \citep[Theorem 3.A.35]{shaked} that convex combinations of an independent and identically distributed (\iid) sample become less dispersed, in terms of the \textit{convex order}, as the weight vector becomes more balanced, in terms of \textit{majorization}. This implies, for instance, that the sample mean of 
size \(k\) is less dispersed than the sample mean of size \(k-1\) (see 
Example~3.A.29 and Corollary~3.A.31 in \cite{shaked}). The existence and finiteness of the mathematical expectation of the underlying random variable is crucial for these characterisations.

However, as discussed in some recent literature, when the mean is not finite, the behaviour is less predictable, and often unexpected. Convex-order comparisons can be extended in certain ways to random variables without finite means \citep{CoteWang2026}, but their interpretation changes substantially in this setting. 
Here we focus on a different phenomenon, formulated in terms of the usual stochastic order.
In particular, \cite{superPareto2024} proved that, for a specific family of distributions obtained by convex transformations of Pareto random variables, named \textit{super-Pareto}, any convex combination of the sample stochastically dominates the generating random variable.
This result has many important applications, including the surprising fact that the sample mean does not always shrink towards a constant when the sample size increases, but, on the contrary, it can get stochastically larger. More recently, a number of contributions enlarged the family of distributions for which such a result holds \citep{ChenShneer2025,ArabConvex2025,Muller2024,Vincent2025}.  
In particular, \cite{Muller2024} considered the class of variables that are convex transformations of the Cauchy distribution, and \cite{ArabConvex2025} proved that the result holds for {the class of }random variables whose inverted distribution {is subadditive}, called \InvSub. These two latter classes seem to be largest known families of distributions for which convex combinations of a sample stochastically dominate the underlying variable, and there is no inclusion between them. 

Taken together, these results naturally raise the question of how stochastic dominance behaves when comparing different linear combinations of a sample, possibly involving distinct weight vectors and sample sizes. A first answer to this question was recently given by \cite{Chen2025Combinations}, who proved that the stochastic dominance holds when the weight vectors satisfy a majorization relation, provided that the inverted distribution of the underlying random variable is concave. This class is strictly contained in the \InvSub\ class of \cite{ArabConvex2025}. The result of \cite{Chen2025Combinations} implies many important results, including --within their class of distributions-- stochastic monotonicity of the sample mean in the sample size. {The range of applications discussed in \cite{superPareto2024}, \cite{ChenShneer2025}, \cite{Muller2024}, \cite{Vincent2025}—spanning economics, finance, insurance, and risk theory—suggests that this emerging literature is likely to have methodological and applied relevance.}

Overall, the recent literature proposed different classes, not always included among them, yielding stochastic dominance properties at different levels of generality. In fact, all the aforementioned results only provide sufficient conditions for the dominance property. It is not difficult to find examples of distributions that are not included in any of these classes but for which, at least for specific choices of the combination coefficients, the dominance result still holds.
Moreover, in empirical risk analysis, one typically observes data from the underlying distribution rather than the distribution itself. Therefore, it is crucial to infer from the data whether a dominance relation of this type is supported empirically.
This could be done with two different approaches, either testing whether the underlying distribution belongs to one of the aforementioned classes, or testing the property directly. 
With regard to the first approach, testing whether a distribution belongs to the super-Pareto, or to the super-Cauchy class, could be relatively straightforward using existing methods, such as the one recently proposed by \cite{landobenjrada}. Testing the condition proposed by \cite{Chen2025Combinations}, which is based on concavity, would be also quite simple.
{On the other hand, developing tests for the subadditivity-based conditions considered in \cite{ChenShneer2025} and \cite{ArabConvex2025} appears to be an open and potentially challenging problem. Overall, the first approach offers no unique prescription: the choice of which condition to test is not obvious and may strongly affect the resulting inference}. Moreover, testing these sufficient conditions would still leave out many interesting cases, as discussed earlier. 
Therefore, the second approach seems to be more direct and effective. One can directly develop tests that detect the dominance between a linear combination of a sample and the underlying variable, or, more generally, between two linear combinations of interest. 

The objective of this paper is to introduce a statistical test to detect dominance between pairs of linear combinations, in absence of any distributional assumption. In Section~\ref{sec:preliminaries}, we review some relevant results and compare the classes of distributions for which these results hold.
Since relying on classes of distributions is not always possible, in Section~\ref{sec:dominance-without} we focus on those properties that can be obtained without any distributional assumptions. Namely, we consider a dominance between a pair of linear combinations as the starting point for obtaining infinitely many other dominance relations, just by relying on the mathematical properties of stochastic dominance. 
Our results appear particularly simple when the linear combinations reduce to the sample mean. In Section~\ref{sec:estimator} we study inference for the distribution of a linear combination {based on iterations of a weighted convolution operator, applied to the empirical distribution}. We establish the uniform almost sure convergence of {such an} estimator and the weak convergence of the corresponding empirical process. 
As simpler special cases, we consider the case when the underlying distribution is a Cauchy, as any convex combination of \iid\ standard Cauchy random variables is again standard Cauchy; and the case when the combination of interest is a sample mean. In Section~\ref{sec:testing}, we propose two versions of a test for the general null hypothesis of dominance between a pair of linear combinations, versus the alternative of non-dominance, with two alternative methods to obtain critical values. 
The first approach is based on the already mentioned closure property of the Cauchy distribution, which can be seen as the least favourable case under the null. The second approach relies on a bootstrap procedure. For both tests, we establish asymptotic validity under the null hypothesis and consistency under the alternative. Unlike standard dominance testing \citep{whang2019}, our problem is one-sample: both distributions being compared are transformations of the same unknown law, leading to a nonstandard Gaussian limit with covariance expressed through a weighted convolution operator. In Section~\ref{sec:simulations}, we examine and compare the two testing approaches by simulating under different scenarios. Finally, in Section~\ref{sec:application} we consider an application to insurance claim-severity data. The analysis shows how the proposed tests can be used in a real-data setting to study whether pooling produces averaged losses that are stochastically larger than individual losses. All proofs are presented in the Appendix.

{\section{Notation and preliminaries}\label{sec:preliminaries}
We now introduce general concepts and terminology to be used throughout. 
Let $X$ be a random variable with cumulative distribution function (CDF) $F$ and survival function $\Fbar=1-F$. The convolution product is denoted as $*$. As usual, the sample mean of an \iid\ sample $X_1,\ldots,X_s$ is denoted as $\xbar{s}$. Recall that, given a nondecreasing and nonnegative function $g$ such that $g(0)=0$, we say that $g$ is anti-starshaped (at the origin) if, for every $x\geq0$ and every $\lambda\in[0,1]$, $g(\lambda x) \geq \lambda\, g(x)$, and that $g$ is subadditive if,  for every $x,y\geq 0$, $g(x+y)\leq g(x) + g(y)$. 
It is well known that concavity implies anti-starshapedness, which, in turn, implies subadditivity. Note that the terms ``increasing'' and ``decreasing'' are taken as ``non-decreasing'' and ``non-increasing''. Still, with respect to notations, we denote with $\to_p$, $\to_{a.s.}$, and $\to_d$ convergence in probability, almost surely, and in distribution of sequences of random variables, respectively. When referring to sequences of stochastic processes, weak convergence is denoted by $\rightsquigarrow$, as in \cite{vw}. Finally, we denote by $\ell^\infty(\R)$ the space of all bounded real-valued functions on $\R$, equipped with the supremum norm{, denoted as usual by $\|\cdot\|_\infty$}.}

{Our main focus is on stochastic dominance. In particular, we study the conditions under which linear combinations satisfy this relation, as well as how to test whether stochastic dominance holds. Since this concept is central to our analysis, we recall its definition.}
\begin{definition}[\cite{shaked}, p. 3]
\label{def:order1}
Given two random variables $X$ and $Y$, we say that $Y$ stochastically dominates $X$, denoted as $Y\geq_{st} X$, if $P(X\leq x)\geq P(Y\leq x)$ for every $x\in\mathbb{R}$.
\end{definition}
Now, to be more specific, given $X_1,X_2,\ldots$ independent and identically distributed with $X$, and vectors of nonnegative weights $\thetabar=(\theta_1,\ldots,\theta_s)$ and $\etabar=(\eta_1,\ldots,\eta_t)$, it can happen that
\begin{equation}\label{SD-cx}
	\theta_1X_1+\cdots+\theta_sX_s\ \geq_{st}\ \eta_1X_1+\cdots+\eta_tX_t.
\end{equation}
\cite{Chen2025Combinations} recently characterised sufficient conditions on the
underlying distribution and on the weight coefficients that imply this result.
Intuitively, a linear combination in which the weight vector is more concentrated in a single component is expected to be closer, in distribution, to the underlying variable, $X$. 
On the other hand, if the weight vector is more balanced, the linear combination may behave differently, for example, shrinking towards a constant, when the mean is finite, or becoming stochastically larger, when the mean is not finite, as discussed above. 
The translation of this intuitive behaviour of a vector's components into a mathematical formulation is possible through the notion of majorization \citep{MarshallOlkinMajor}.

{As a particularly relevant special case of (\ref{SD-cx}), taking $\etabar=(1,0,...,0)$, the dominance 
\begin{equation}
\label{SD-2}
\theta_1X_1+\cdots+\theta_sX_s\geq_{st}X, 
\end{equation}
for every integer $n\geq 2$ and every choice of nonnegative real numbers $\theta_1,\ldots,\theta_n$ that sum up to 1, has been studied in \cite{superPareto2024}, \cite{Muller2024}, \cite{ChenShneer2025} and \cite{ArabConvex2025}. These authors identify different classes of distributions for which the dominance (\ref{SD-2}) holds.
Among these classes, the largest ones are those proposed by \cite{Muller2024} and \cite{ArabConvex2025}. There is no inclusion between these two classes, in fact, the former one includes distributions on the entire real line, whereas the latter allows for more flexibility in terms of shape, such as discontinuities. In particular, the class of \cite{ArabConvex2025} is characterised by the subadditivity of the inverted distribution, {i.e., the} CDF $H(x)=1-F(\tfrac1x)$.} 



{The stochastic dominance (\ref{SD-cx}) needs some control on the weight coefficients, as mentioned above. We recall the relevant majorization notion between vectors in $\R^s$.}
\begin{definition}(\cite{MarshallOlkinMajor}, Definition~1.A.1)
	\label{def:maj}
	Let $\thetabar=(\theta_1,\ldots,\theta_s),\,\etabar=(\eta_1,\ldots,\eta_s)\in\R^s$.
	We say that $\thetabar$ is dominated by $\etabar$ in majorization order, denoted $\thetabar\prec\etabar$ if
	$\sum_{i=1}^k\theta_{i:s}\geq\sum_{i=1}^k\eta_{i:s},\;k=1,\ldots,s-1$, and $\sum_{i=1}^s\theta_i=\sum_{i=1}^s\eta_i$,
	where $\theta_{i:s}$ and $\eta_{i:s}$ represent the coordinates of $\thetabar$ and $\etabar$, respectively, ordered increasingly.
\end{definition}
The relation $\thetabar\prec\etabar$ generally means that the components of $\thetabar$ are more spread out compared to those of $\etabar$.
%

%
\cite{Chen2025Combinations} proved a simple sufficient condition for (\ref{SD-cx}) to hold.
\begin{theorem}[\cite{Chen2025Combinations}, Theorem~1]\label{theorem chen}
Let $X$ be a nonnegative random variable with CDF $F$ such that $H(x)=1-F(\tfrac1x)$ is concave, and $X_1,\ldots,X_s$ \iid\ with $X$. If $\thetabar,\,\etabar\in[0,+\infty)^s$ are such that $\thetabar\prec\etabar$, then $\sum_{i=1}^s\theta_i X_i\geq_{st}\sum_{i=1}^s\eta_i X_i$.
\end{theorem}
Note that this result still implies the stochastic dominance (\ref{SD-2}). This same conclusion holds under the weaker assumption of subadditivity of $H(x)$, as proved in Theorem~3.1 in \cite{ArabConvex2025}. Hence, one may naturally wonder whether this latter and weaker assumption suffices to obtain (\ref{SD-cx}). The answer is negative. The following example shows that not even anti-starshapedness is enough.
\begin{example}
Consider the CDF defined as
$$
F(x)=\left\{\begin{array}{ccl}
0, & \qquad & \mbox{if }\,x\leq\frac14, \\[1.5ex]
0.8-\frac{0.2}{x}, & & \mbox{if }\,\frac14\leq x<\frac12, \\[1.5ex]
0.4, & & \mbox{if }\,\frac12\leq x<1, \\[1.5ex]
1-\frac{0.6}{x}, & & \mbox{if }\,x\geq 1.
\end{array}\right.
$$
It is easily seen that the inverted CDF $H(x)=1-F(\tfrac1x)$ is anti-starshaped, but not concave.
Considering, for example, the weight vectors $(0.2,\,0.8)\prec(0.3,\,0.7)$, one can verify that the corresponding CDFs cross, so {stochastic dominance does not hold}.
\end{example}
	\begin{remark}
\label{rem:size}
Note that majorization assumes that $\thetabar$ and $\etabar$ have the same number of components. {This restriction is easily overcome}. Indeed, every pair of vectors can be compared by filling with zeros the missing components of the shorter one. Take $\thetabar\in\R^s$ and $\etabar\in \R^t$, where $s>t$. Then, define the vector $\widetilde\eta = (\eta_1,\dots,\eta_t,\widetilde{0}_{s-t})\in\mathbb{R}^{s}$, where $\widetilde{0}_{s-t}$ is a vector with $s-t$ coordinates all equal to 0. In this case, for example, we can compare the vector $(\frac12,\frac12)$ with the $s$-dimensional vector $(\frac1s,\ldots,\frac1s)$, for $s>2$, as $(\frac12,\frac12,\widetilde{0}_{s-2})\prec(\frac1s,\ldots,\frac1s)$. 
\end{remark}

According to Remark~\ref{rem:size}, we may also compare linear combinations where the weight vectors have a different number of coordinates. This is particularly interesting as it allows to stochastically compare sample means of different sizes. Specifically, this means that, if $H(x)=1-F(\tfrac1x)$ is concave, the sample mean is stochastically monotone with respect to the sample size. 
In particular, among these sample means, the case $s=2$ is stochastically the smallest; hence it is the most stringent benchmark and a natural candidate for statistical testing. This idea allows for some hierarchical dominance results, and will be explored later in Section~\ref{sec:dominance-without}.

We conclude the section with a remark.
\begin{remark}
	The sufficient condition in Theorem~\ref{theorem chen} is formulated in terms of majorization, which is only a partial order. When $s=t=2$, majorization becomes a total order (for vectors with the same sum). In higher dimensions, incomparability is common—many pairs of vectors cannot be ordered—so the practical scope of the theorem is more limited. This does not mean that stochastic dominance cannot hold when weights are incomparable; rather, the sufficient condition is then silent and does not indicate whether dominance holds, nor in which direction. This provides additional motivation for a direct statistical verification of condition~(\ref{SD-cx}), which we pursue in Section~\ref{sec:testing}.
\end{remark}

\section{Stochastic dominance without distributional assumptions}
\label{sec:dominance-without}
{We have been interested in finding conditions {on the distribution of the underlying variables and the weights such that} (\ref{SD-cx}) holds}. However, this stochastic dominance can also be seen as a starting point to obtain infinitely many other dominance relations, using somewhat standard arguments in ordering theory. In this section, we will focus on the implications of (\ref{SD-cx}), focusing on sample means, indisputably the most important linear combinations based on a sample. 
We show that sample means satisfy a hierarchical property, whenever $\xbar{s}\geq_{st}\xbar{t}$. Namely, the latter condition implies $\xbar{ks}\geq_{st} \xbar{kt}$ for every positive integer $k$. This result can be extended using general linear combinations instead of just sample means. Moreover, we search for a result that can be seen as a special case of Theorem~\ref{theorem chen}, but does not rely on distributional assumptions on the underlying random variables. 
Such a result is possible if we assume $\xbar{s}\geq_{st} X$ plus a stronger type of ordering between the weight vectors, which can be seen as a special case of majorization, defined as follows. 
\begin{definition}
Let $\thetabar=(\theta_1,\ldots,\theta_k)$ be a vector of nonnegative weights such that $\sum_{i=1}^k \theta_i=1$ and $h\geq2$ an integer.
    \begin{enumerate}
    \item
    An $h$-split of $\thetabar$ is a $(k+h-1)$-dimensional vector $(\theta_1,\ldots,\theta_j^{(h)},\ldots,\theta_k)$ where, for some $j\in\{1,\ldots,k\}$, the coordinate $\theta_j$ is replaced by the $h$-dimensional vector $\theta_j^{(h)}=(\frac{\theta_j}{h},\cdots,\frac{\theta_j}{h})$.
    \item
    Given vectors $\thetabar$ and $\etabar$ with nonnegative components we say that $\thetabar$ is $h$-split majorized by $\etabar$, represented by $\thetabar\prec_s\etabar$, if $\thetabar$ can be obtained from $\etabar$ by a finite number of $h$-splits.
    \end{enumerate}
\end{definition}

It is immediate that $h$-split majorization implies the usual majorization, as given in Definition~\ref{def:maj}, after filling with zeros the missing components of the shorter vector.
Indeed, if $\thetabar$ is obtained from $\etabar\in\mathbb{R}^k$ by a finite number of $h$-splits, then
$\thetabar\in\mathbb{R}^{\ell}$ for some $\ell> k$. 
{Then $\thetabar \prec_h \etabar$ implies $\thetabar \prec \widetilde\eta$, where $\widetilde\eta = (\eta_1,\dots,\eta_k,\widetilde{0}_{\ell-k})\in\mathbb{R}^{\ell}$ is defined as in Remark \ref{rem:size}}.

We are now ready to state the main result of this section. Let $\otimes$ denote the Kronecker product between vectors.

\begin{proposition} \label{t1} \
    \begin{enumerate}
	\item 
    Consider weight vectors $\thetabar\in\mathbb{R}^s$ and $\etabar\in\mathbb{R}^t$. If $\sum_i\theta_iX_i\geq_{st}\sum_j\eta_jX_j$ then, for every vector $\overline{\pi}\in\mathbb{R}^k$,
    $$
    \sum_{i=1}^{sk} (\thetabar\otimes \overline{\pi})_iX_i\geq_{st}\sum_{j=1}^{tk}(\etabar\otimes \overline{\pi})_jX_j.
    $$ 
	\item 
    If $\xbar{s}\ge_{st} \xbar{t}$, then, for every positive integer $k$, $\xbar{ks}\ge_{st} \xbar{kt}$.
	\item 
    If $\xbar{s}\geq_{st} X$ and $\thetabar\prec_s \etabar$ with nonnegative coordinates, then $\sum_i\theta_iX_i\geq_{st}\sum_j\eta_jX_j$.
	\item 
    If $\xbar{s}\ge_{st} \xbar{t}$, then, for every positive integer $k$,
	$$
	\xbar{s+k}\geq_{st} \frac{s}{s+k}\,\xbar{t}+\frac{k}{s+k}\,	\xbar{k}.
	$$
	\end{enumerate}
\end{proposition}

\begin{figure}[h]
	\centering
	\begin{tikzpicture}[
		node/.style={draw, rounded corners, inner sep=3pt, font=\small},
		arr/.style={-{Stealth[length=2mm]}, line width=0.5pt}
		]
		
		\def\N{24}         
		\def\Cols{6}       
		\def\xstep{2.0}    
		\def\ystep{1.6}    
		
		\foreach \i in {1,...,\N} {
			\pgfmathtruncatemacro{\col}{mod(\i-1,\Cols)}
			\pgfmathtruncatemacro{\row}{(\i-1)/\Cols}
			\pgfmathsetmacro{\x}{\col*\xstep}
			\pgfmathsetmacro{\y}{-\row*\ystep}
			
			\ifnum\i=1
			\node[node] (x\i) at (\x,\y) {$X$};
			\else
			\node[node] (x\i) at (\x,\y) {$\xbar{\i}$};
			\fi
		}
		\draw[arr] (x2) -- (x1);
        \draw[arr] (x3) -- (x2);
        \draw[arr] (x4) to [bend right=20] (x2);
        \draw[arr] (x6) to [bend right=20] (x4); \draw[arr] (x6) to [bend right=25] (x3);
        \draw[arr] (x8) -- (x4);
        \draw[arr] (x9) -- (x3); \draw[arr] (x9) -- (x6);
        \draw[arr] (x10) -- (x5);
        \draw[arr] (x12) -- (x4); \draw[arr] (x12) -- (x6);
        \draw[arr] (x14) -- (x7);
        \draw[arr] (x15) -- (x10); 
        \draw[arr] (x16) -- (x8);
        \draw[arr] (x18) -- (x9); \draw[arr] (x18) -- (x12); 
        \draw[arr] (x20) to [bend left=20] (x10);
        \draw[arr] (x21) -- (x14);
        \draw[arr] (x22) -- (x11);
        \draw[arr] (x24) -- (x16); \draw[arr] (x24) to [bend right=30] (x12);
	\end{tikzpicture}
	\caption{Dominance relations implied by $\xbar{3}\geq_{st} \xbar{2} \geq_{st} X,$ with sample means of size $n\le 24$.}
	\label{fig:dominance-network-24}
\end{figure}

A description of {some} dominance relations that can be derived from Proposition~\ref{t1}, given some starting points, is given in Figure~\ref{fig:dominance-network-24}.

It is worth  noting that Proposition~\ref{t1} only proves strict implications. Even in the case of the sample mean, in general, $\xbar{s}\geq_{st}\xbar{t}$ does not imply that  $\xbar{k}\geq_{st}\xbar{t}$, for $k>s$. From Example~2.1 in \cite{Vincent2025}, it can be seen that the distribution of the St.\ Petersbourg lottery, represented as $X\stackrel{d}{=}2^Y$, where $Y$ is a geometric random variable, satisfies $\xbar{2} \geq_{st}X$ but $\xbar{3} \not\geq_{st}X$. This result can be proved applying Theorem 4.2 of \cite{Vincent2025}.

{

\section{Estimating the distribution of a linear combination}
\label{sec:estimator}
We start now paving the way to tackle the problem of testing if condition (\ref{SD-cx}) holds, which will be done in Section~\ref{sec:testing}. To perform inference on this stochastic dominance relation, we first need an estimator of the distribution of weighted sums of \iid\ random variables. This estimator can be obtained by applying the plug-in method to an iteration of the weighted convolution operator, that extends the usual convolution, defined as follows.

\begin{definition}
\label{def:wei-conv-op}
Let $\theta_1$ and $\theta_2$ be nonnegative real numbers and $F$ and $G$ two CDFs. We define the weighted-convolution {operator as}
\begin{equation}
\label{eq:wei-conv-op}
L_{\theta_1,\theta_2,F}(G)(x)=\int G(\tfrac{x-\theta_2t}{\theta_1})\,dF(t),\qquad x\in\mathbb{R}.
\end{equation}
\end{definition} 
This operator reduces to the usual convolution when taking $\theta_1=\theta_2=1$, hence it defines an extension of the convolution product. Moreover, it is easily seen that, if $X_1$ and $X_2$ are \iid\ with CDF $F$, $L_{\theta_1,\theta_2,F}(F)(x)=P(\theta_1 X_1+\theta_2 X_2\leq x)$. For higher dimensional sums, a suitable iterated application of the operator gives a characterisation of the distribution function, as it is easily verified that, given $X_1,\ldots,X_s$ \iid\ with CDF $F$ and weights $\thetabar=(\theta_1,\ldots\theta_s)$,
\begin{equation}
\label{eq:sum-n}
\begin{array}{rcl}
F_{\thetabar}(x) & = & P(\theta_1{X}_1+\cdots+\theta_s{X}_s\leq x) \\[1ex]
 & = & L_{1,\theta_s,F}\circ L_{1,\theta_{s-1},F}\circ\cdots\circ L_{1,\theta_3,F}\circ L_{\theta_1,\theta_2,F}(F)(x)
=: L_{\thetabar}(F).
\end{array}
\end{equation}
Given a random sample $\widetilde X_1,\ldots,\widetilde X_n$ of size $n$ from $X$, the empirical CDF is denoted, as usual, with $\F_n(x)=\frac1n\sum_{i=1}^n\mathbb{I}_{\{\widetilde X_i\leq x\}}$, hence the plug-in estimator of {$F_{\thetabar}$} is
\begin{equation}
\label{eq:sum-n-plug}
\mathbb{F}_{n,\thetabar}(x)=L_{1,\theta_s,\mathbb{F}_n}\circ L_{1,\theta_{s-1},\mathbb{F}_n}\circ\cdots\circ L_{1,\theta_3,\mathbb{F}_n}\circ L_{\theta_1,\theta_2,\mathbb{F}_n}(\mathbb{F}_n)(x)=L_{\thetabar}(\mathbb{F}_n)(x).
\end{equation}
Note that, for each fixed $x$, this estimator can be seen as a $V$-statistic of order $s$ \citep[Ch. 3.5.3]{shao}, that is,
$$ 
\mathbb{F}_{n,\thetabar}(x)=\frac{1}{n^s}
\sum_{i_1=1}^n\!\cdots\!\sum_{i_s=1}^n
\mathbb I_{\{\theta_1u_1+\cdots+\theta_su_s\le x\}}.
$$
In what follows, however, we use the weighted-convolution representation, which is more convenient for our purposes and allows us to derive the asymptotic and bootstrap results through empirical-process theory.

We now address the asymptotic properties of the estimator $\F_{n,\thetabar}$. The following result establishes the almost sure uniform convergence of $\F_{n,\thetabar}$ to $F_{\thetabar}$.
\begin{proposition}
\label{prop:unif-conv}
For every $\thetabar=(\theta_1,\ldots,\theta_s)\in(0,+\infty)^s$, $\lim_{n\rightarrow+\infty}\left\| \mathbb{F}_{n,\thetabar}-\mathbb{F}_{\thetabar}\right\|_\infty=0$ almost surely.
\end{proposition}
Besides the almost sure convergence of {$F_{n,\thetabar}$}, we need to describe the weak convergence of the associated empirical process. For this purpose, the Hadamard differentiability of {$L_{\thetabar}$} is an essential tool.
\begin{proposition}
	\label{prop:Hadamard}
	Given a vector of weights $\thetabar=(\theta_1,\ldots,\theta_s){\in\mathbb{R}^s}$, a CDF $F$, and $h\in \ell^\infty(\mathbb{R})$, {the Hadamard derivative of $L_{\thetabar}$ at $F$ in the direction $h$ is given by}
	\begin{equation}
		\label{eq:hadamard}
		L_{\thetabar}'(F;h)=\sum_{(a_1,\ldots,a_s)\in
			{D}} L_{1,\theta_s,a_s}\circ L_{1,\theta_{s-1},a_{s-1}}\circ\cdots\circ L_{1,\theta_3,a_3}\circ L_{\theta_1,\theta_2,a_1},
	\end{equation}
	where $
	D=\{(h,F,\ldots,F),\,(F,h,F,\ldots,F),\ldots,(F,\ldots,F,h)\}$.
\end{proposition}
Next, we prove the weak convergence of the empirical process associated with $\mathbb{F}_{n,\thetabar}$, namely,  $\sqrt{n}\,\left(\mathbb{F}_{n,\thetabar}-F_{\thetabar}\right)$, and characterise the main properties of the limiting process. 
\begin{theorem}
\label{thm:weak-Fthetabar}
Let $F$ be a continuous CDF. Then, for every $\thetabar=(\theta_1,\ldots,\theta_s)\in(0,+\infty)^s$, there exists a 
centred Gaussian process $\mathbb{H}_{\thetabar}^F$ such that 
    \begin{enumerate}
    \item 
    when $n\rightarrow+\infty$,
    $\sqrt{n}\,\left(\mathbb{F}_{n,\thetabar}-F_{\thetabar}\right)\rightsquigarrow\mathbb{H}_{\thetabar}^F=L_{\thetabar}'(F;\mathbb{B}_F)$,
    where $\mathbb{B}$ is a Brownian bridge, $\BB=\mathbb B\circ F$, and $L_{\thetabar}'(F;\BB)$ is the Hadamard derivative of $L_{\thetabar}$ at $F$ in the direction of $\BB$, given by (\ref{eq:hadamard}).
    \item
    let $\thetabar_{-j}=(\theta_1,\ldots,\theta_{j-1},\theta_{j+1},\ldots,\theta_s)$ be the {$(s-1)$-dimensional} vector obtained by removing the $j$-th coordinate, then
    \begin{equation}
    \label{eq:procH}
    \mathbb{H}_{\thetabar}^F=\sum_{j=1}^s  \int \BB\left(\tfrac{x-t}{\theta_j}\right)\, dF_{\thetabar_{-j}}(t)\ =\sum_{j=1}^s  \BB\left(\tfrac{\cdot}{\theta_j}\right) * F_{\thetabar_{-j}}.
    \end{equation}
    \item
    the covariance operator of $\mathbb{H}_{\thetabar}^F$ is given by
    $\operatorname{Cov}\left(\mathbb{H}_{\thetabar}^F(x),\mathbb{H}_{\thetabar}^F(y)\right)=\sum_{j,k}\Omega_{j,k}(x,y)$, where
 $$
    \Omega_{j,k}(x,y)=\left\{\begin{array}{lcl}
    	\displaystyle
    	\int F_{\thetabar_{-j}}(x-\theta_ju)F_{\thetabar_{-k}}(y-\theta_ku)\,dF(u)  - F_{\thetabar}(x)F_{\thetabar}(y),
    	& \quad & \mbox{if }j\not=k, \\[4ex]
    	\displaystyle
    	\int F_{\thetabar_{-j}}(x\wedge y-\theta_ju)\,dF(u) - F_{\thetabar}(x)F_{\thetabar}(y),
    	& \quad & \mbox{if }j=k.
    \end{array}\right.
    $$
    
    \end{enumerate}
\end{theorem}

\subsection{The case of Cauchy distributed variables}
\label{sub:Cauchy}
Let us denote by $\mathcal{C}(x)=\frac{1}{\pi }\arctan(x)+\frac{1}{2}$, $x\in \R$, the standard Cauchy CDF. This special case is of particular interest, as linear combinations of \iid\ such variables follow a scaled Cauchy distribution.
{In particular, when the components of the weight vector sum up to 1, the resulting (convex) combination still has distribution $\mathcal C$.} 
Hence, we may find a simplified representation of the limiting process $\mathcal{H}_{\thetabar}^{\mathcal{C}}$. Assume then that $X_1,\ldots,X_s$ are \iid\ with CDF $\mathcal{C}$, and the weights $\thetabar=(\theta_1,\ldots,\theta_s)\in(0,+\infty)^s$ are such that $\|\thetabar\|=1$, for simplicity. 
Therefore,
$\mathcal{C}_{\thetabar_{-j}}(x)=P\left(\sum_{i\not=j}\theta_iX_i\leq x\right)=\mathcal{C}\left(\tfrac{x}{1-\theta_j}\right)$. We may state a specific version of Theorem~\ref{thm:weak-Fthetabar}, whose proof follows immediately from the relation between $\mathcal{C}_{\thetabar_{-j}}$ and $\mathcal{C}$.
\begin{corollary}
\label{cor:emp-Cauchy}
Let $F=\mathcal C.$ Then
\begin{enumerate}
\item
the limit process derived in Theorem~\ref{thm:weak-Fthetabar} reduces to
\begin{equation}
\label{eq:procC}
\mathbb{H}_{\thetabar}^{\mathcal{C}}=\sum_{j=1}^s  \int {\mathbb B_{\mathcal C}}\left(\tfrac{x-(1-\theta_j)t}{\theta_j}\right)\, d\mathcal{C}(t)\ = \sum_{j=1}^s  \mathbb B_{\mathcal C}\left(\tfrac{\cdot}{\theta_j}\right)\, *\mathcal{C}(\tfrac{\cdot}{1-\theta_j});
\end{equation}
\item
the covariance operator of $\mathbb{H}_{\thetabar}^{\mathcal{C}}$ is given by
$\operatorname{Cov}\left(\mathbb{H}_{\thetabar}^{\mathcal{C}}(x),\mathbb{H}_{\thetabar}^{\mathcal{C}}(y)\right)=\sum_{j,k}\Omega_{j,k}^{\mathcal{C}}(x,y)$, where
 $$
    \Omega_{j,k}^{\mathcal{C}}(x,y)=\left\{\begin{array}{lcl}
    	\displaystyle
    	\begin{array}{l}
    		\displaystyle\int \mathcal{C}\left(\tfrac{x-\theta_ju}{1-\theta_j}\right)\mathcal{C}\left(\tfrac{y-\theta_ku}{1-\theta_k}\right)\,d\mathcal{C}(u) - \mathcal C(x)\mathcal C(y),
    	\end{array}
    	& \qquad & \mbox{if }j\not=k, \\[4ex]
    	\displaystyle
    	\int \mathcal{C}\left(\tfrac{x\wedge y-\theta_ju}{1-\theta_j}\right)\,d\mathcal{C}(u) -  \mathcal C(x)\mathcal C(y),
    	& \qquad & \mbox{if }j=k.
    \end{array}\right.
    $$
\end{enumerate}
\end{corollary}

\subsection{A further special case: sample means}
Sample means play an important role in statistics, so we give a closer look to the representations corresponding to this case, where $\thetabar=(\tfrac1s,\ldots,\tfrac1s)$. As the weights are now all equal, it is natural to seek expressions where integration is taken with respect to the usual convolution powers. We introduce some specific notation. The convolution power, as mentioned after Definition~\ref{def:wei-conv-op}, is $\conv{F}{2}=F\ast F=L_{1,1,F}(F)$, and it is easily seen that $\conv{F}{s}=L_{1,1,F}(\conv{F}{(s-1)})$. As for the Cauchy case, we start be characterising the distributions of the form F$_{\thetabar_{-j}}$, to see how this reflects on the further integrations. Taking into account that the coordinates of $\thetabar$ are equal to $\frac1s$, considering $X_1,\ldots,X_s$ \iid\ with CDF $F$, we have
$$F_{\thetabar_{-j}}(x)=F_{\xbar{s-1}}\left(\tfrac{sx}{s-1}\right).$$
Rewrite
$$
\int\BB\left(\tfrac{x-t}{\theta_j}\right)\,dF_{\thetabar_{-j}}(t)=\int\BB\left(\tfrac{x-t}{1/s}\right)\,dF_{\thetabar_{-j}}(t)=\int\BB\left(sx-u\right)\,dF_{\thetabar_{-j},s}(u),
$$
where $F_{\thetabar_{-j},s}(u)=F_{\thetabar_{-j}}(\frac{u}{s})=F_{\xbar{s-1}}(\tfrac{x}{s-1})=\conv{F}{(s-1)}(x)$. Hence, recalling (\ref{eq:procH}), the first assertion in following statement is an immediate consequence of this arguing. The second assertion is a consequence of the fact that power convolutions of a Cauchy distribution coincide with the Cauchy distribution. {Finally, a direct translation of the representation of the covariance operator described in Theorem~\ref{thm:weak-Fthetabar} (assertion 3.), taking into account the rewriting of $F_{\thetabar_{-j}}$ as above, yields the third assertion.}
\begin{corollary}
\label{cor:samp-means}
Let $F$ be continuous, $s\geq 2$ be fixed, take $\thetabar=(\tfrac1s,\ldots,\tfrac1s)$, and represent the limit of the corresponding empirical process, according to Theorem~\ref{thm:weak-Fthetabar}, as $\mathbb{H}_s^F$. 
    \begin{enumerate}
    \item
    Then $\mathbb{H}_s^F(x)=s\BB\ast\conv{F}{(s-1)}(sx)$. 
    \item
    In the particular case where $X$ is Cauchy distributed, 
    $$
    \mathbb{H}_s^{\mathcal{C}}(x)=s \int {\mathbb B_{\mathcal C}}\bigl(sx-(s-1)t\bigr)\,d\mathcal C(t).
    $$
    \item
    {The covariance operator of the process is
    \begin{eqnarray*}
    \lefteqn{\operatorname{Cov}\left({\mathbb H_s^F(x),\mathbb H_s^F(y)}\right)} \\[1.5ex]
     & = &
    	s^2\left(\int \conv{F}{(s-1)}(sx-u)\,\conv{F}{(s-1)}(sy-u)\,dF(u)
    	-\conv{F}{s}(sx)\conv{F}{s}(sy)\right).
    \end{eqnarray*}}
    \end{enumerate}
\end{corollary}

\section{Testing dominance between linear combinations}
\label{sec:testing}
                                                                                                                                                                                                                                                 
We consider a test for the null hypothesis} $\mathcal{H}_0: \sum_i\theta_i X_i\geq_{st}\sum_i\eta_i X_i$ {versus the alternative} $\mathcal{H}_1: \sum_i\theta_i X_i\not\geq_{st}\sum_i\eta_i X_i$, for given weight vectors $\thetabar=(\theta_1,\ldots,\theta_s)\in{(0,+\infty)^s}$ and $\etabar=(\eta_1,\ldots,\eta_t)\in{(0,+\infty)^t}$, where $s$ and $t$ are not necessarily equal. {As is customary in the literature on stochastic dominance tests (see Chapter 2.2.2 of \citet{whang2019} for a review), we consider a Kolmogorov–Smirnov-type test statistic, based on the sup-norm of the difference between empirical versions of $F_{\thetabar}$ and $F_{\etabar}$, namely,}
\begin{equation}
\label{eq:test-g}
T_{\thetabar,\etabar}(\F_n)=\left\| \mathbb{F}_{n,\thetabar}-\mathbb{F}_{n,\etabar}  \right\|_\infty,
\end{equation}
where $\mathbb{F}_{n,\thetabar}=L_{\thetabar}(\mathbb{F}_n)$ and $\mathbb{F}_{n,\etabar}=L_{\etabar}(\mathbb{F}_n)$ are given by (\ref{eq:sum-n-plug}). Differently from the classic problem of testing stochastic dominance, in our case the test statistic is a functional of only one empirical CDF, $\F_n$.
This test statistic is clearly location and scale invariant.
                                                                                                                                                                                                                         
We propose two different tests both based on $T_{\thetabar,\etabar}$. The first one, based on the Cauchy distribution, applies to weight vectors such that $\|\thetabar\|=\|\etabar\|=1$, {exploits} the fact that convex combinations of Cauchy variables are also Cauchy distributed, making it possible to identify the asymptotic least favourable distribution of the test statistic under the null hypothesis. 
Actually, this test may be adapted to general weight vectors by suitable scaling, thus at the cost of some extra notational complexity. The second test is based on a bootstrap procedure, and does not depend on any restriction on the weight vectors. {Both tests are consistent and have asymptotically bounded size under the null hypothesis.} 
For convenience, we shall be writing $F\in\mathcal{H}_0$ whenever this hypothesis is satisfied for random variables with CDF $F$, and similarly for other hypothesis under consideration.

\subsection{The test based on the Cauchy distribution}
Recall that for this subsection we will be assuming that $\|\thetabar\|=\|\etabar\|=1$. We have denoted the standard Cauchy CDF as $\mathcal{C}$ (see Subsection~\ref{sub:Cauchy}), so its stability with respect to convex combinations implies that $L_{\thetabar}(\mathcal{C})=\mathcal{C}$. Moreover, it is clear that $C\in\mathcal{H}_0$.
Therefore, taking into account the behaviour of the distribution of convex combinations of \iid\ random variables, one could expect that, if $F=\mathcal C$, $\mathbb{F}_{n,\thetabar}$ and $\mathbb{F}_{n,\etabar}$, {the estimated distributions} of $\sum_i\theta_iX_i$ and $\sum_i\eta_iX_i$, respectively, will be ``close'' to each other, and to $\mathcal C$, especially for large $n$, compared to the general case when $F\in \mathcal H_0$. 
In other words, we expect that the Cauchy case is the least favourable for the distribution of $T_{\thetabar,\etabar}$ under the null hypothesis. 
Focus now on the strict sub-null $\widetilde{\mathcal{H}}_0=\{F:\, F_{\thetabar}(x)<F_{\etabar}(x)\text{ a.e.} \}$, thus excluding boundary cases for which $F_{\thetabar}$ and $F_{\etabar}$ may coincide on a set of non-null Lebesgue measure. 
Note that, if $X$ has an {almost everywhere} analytic density function, which includes {all the standard} families of {continuous} distributions, the coincidence on an interval of the densities implies that they coincide in the whole domain, so this case is excluded from $\mathcal H_0 \setminus \widetilde{\mathcal H}_0$.
Apart from the special case of $\mathcal C$ ({or the trivial case of distributions which coincide to 0 outside their supports, which is irrelevant for our analysis}), we are not aware of any further examples of distributions in $\mathcal H_0 \setminus \widetilde{\mathcal H}_0$.

\begin{theorem}
\label{thm:Htilde}
Assume that $F$ is continuous. 
    \begin{enumerate}
    \item
    If $F\in\widetilde{\mathcal{H}}_0$, then $\sqrt n \left\|\mathbb{F}_{n,\thetabar}-\mathbb{F}_{n,\etabar}\right\|_\infty\to_p 0$.
    \item
    If $F=\mathcal C$, then $\sqrt n \left\|\mathbb{F}_{n,\thetabar}-\mathbb{F}_{n,\etabar}\right\|_\infty\to_p  \left\|\mathbb{H}_{\thetabar}^{\mathcal{C}}-\mathbb{H}_{\etabar}^{\mathcal{C}}\right\|_\infty\geq0$.
    \end{enumerate}
\end{theorem}

We now describe the test based on this particular behaviour of the Cauchy distribution. 

\bigskip

\textbf{The Cauchy based test:}
\textsl{
    \begin{enumerate}
    \item
    Denote with $\mathbb C_n$ the empirical CDF obtained by sampling from the distribution $\mathcal C$.
    \item
    Generate, via the Monte Carlo method, the critical value $c_{\alpha,n}$ as the $(1-\alpha)$-quantile of $\sqrt{n}\,T_{\thetabar,\etabar}(\mathbb{C}_n)$.
    \item
    Reject $\mathcal{H}_0$ if ${\sqrt{n}}\,T_{\thetabar,\etabar}(F_n)\geq c_{\alpha, n}$, where $F_n$ is a realization of $\F_n$.
    \item
    The $p$-value is given by $P(T_{\thetabar,\etabar}(\mathbb C_n)\geq T_{\thetabar,\etabar}(F_n)).$
    \end{enumerate}}
This test has the appropriate properties, as shown in the next statement.
\begin{proposition}
\label{prop:test-conv-Cauchy}
Assume that $F$ is continuous.
    \begin{enumerate}
    \item
    {If $F=\mathcal C$, $\lim_{n\to\infty}P(\text{reject }\mathcal{H}_0)=\alpha.$}
	\item 
     If $F\in\widetilde{\mathcal{H}}_0$, $\lim_{n\to\infty}P(\text{reject }\mathcal{H}_0)=0$. 
	\item 
    If $F\in\mathcal{H}_1$, $\lim_{n\to\infty}P(\text{reject }\mathcal{H}_0)=1$.
	\end{enumerate}
\end{proposition}
Although the boundary case when $F\in \mathcal H_0 \setminus \widetilde{\mathcal H}_0$ seems to be practically negligible, the behaviour of the test in this case could not be established. In fact, the analysis would require a much finer control of the covariance structure of the limiting Gaussian process, which we are currently unable to obtain. For this reason, we develop and study a second approach based on the bootstrap.

\subsection{A bootstrap test}
{We now propose an alternative test, where critical values are computed via bootstrap, in the spirit of \cite{barrett2003} (see also \cite{lando2025new} or \cite{barrett2014consistent}).}
Under the null hypothesis, $F_{\thetabar}(x)-F_{\etabar}(x)\leq0$ for every $x\in\R$. Therefore,
\begin{equation}
\label{eq:upper}
\mathbb{F}_{n,\thetabar}(x)-\mathbb{F}_{n,\etabar}(x)\leq \left(\mathbb{F}_{n,\thetabar}(x)-\mathbb{F}_{n,\etabar}(x)\right)-\left(F_{\thetabar}(x)-F_{\etabar}(x)\right),
\end{equation}
which implies that
\begin{equation}
\label{eq:upper1}
\sqrt{n}\,T_{\thetabar,\etabar}(\F_n)\leq \sqrt{n}\left\|\left(\left(\F_{n,\thetabar}(x)-\F_{n,\etabar}(x)\right)-\left(F_{\thetabar}(x)-F_{\etabar}(x)\right)\right)_+  \right\|_\infty\to_d \left\| \mathbb{H}_{\thetabar}^F-\mathbb{H}_{\etabar}^F\right\|_\infty.
\end{equation}
The rejection region for this test will be based in the distribution of $\left\|\mathbb{H}_{\thetabar}^F-\mathbb{H}_{\etabar}^F\right\|_\infty$, thus defining a conservative test. For practical matters, the distribution of $\left\|\mathbb{H}_{\thetabar}^F-\mathbb{H}_{\etabar}^F\right\|_\infty$ can be approximated by bootstrap.
Before describing the test, we justify the bootstrapping procedure. Let $\F_n^b(x)=\frac1n\sum_{i=1}^n M_i\mathbb{I}_{\{\widetilde X_i\leq x\}}$, where $M_i$ is the number of times $\widetilde X_i$ is resampled from the original sample, be the bootstrap estimator of $\F_n$. 
Correspondingly, we can compute $\F_{n,\thetabar}^b$ and define the bootstrap version of the bounding process in (\ref{eq:upper})
$$\mathcal{G}_n=\sqrt{n}\left((\F_{n,\thetabar}-\F_{n,\etabar})-(F_{\thetabar}-F_{\etabar})\right)$$
with 
$$
\mathcal{G}_n^b=\sqrt{n}\left(\left(\F_{n,\thetabar}^b-\F_{n,\etabar}^b\right)
    -\left(\F_{n,\thetabar}-\F_{n,\etabar}\right)\right).
$$
The functional delta method for the bootstrap implies that these two processes have the same limit:	$\mathcal{G}_n^b\rightsquigarrow \mathbb{H}_{\thetabar}^F-\mathbb{H}_{\etabar}^F$, which then implies that $\left\| \mathcal{G}_n^b\right\|_\infty\to_d \left\|{\mathbb{H}_{\thetabar}^F-\mathbb{H}_{\etabar}^F}\right\|_\infty$.

\bigskip

\textbf{The bootstrap test}
\textsl{
    \begin{enumerate}
    \item
    Build $b=1,\ldots,B$ resamples from the empirical CDF $F_n$ of the original sample, getting the corresponding empirical distribution functions $F_n^b$, as described above.
    \item
    Get the bootstrapped versions of the distribution of the weighted sums $F_{n,\thetabar}^b$ and $F_{n,\etabar}^b$.
    \item
    The bootstrap $p$-value, defined as $p=P\left(\left\|\mathcal{G}_n^b\right\|_\infty>\sqrt{n}\,T_{\thetabar,\etabar}(F_n)\right)$, is approximated by
    $$
    \widehat{p}= \frac1B\sum_{b=1}^B \mathbb{I}\left(\left\|\mathcal{G}_n^b\right\|_\infty>\sqrt{n}\,T_{\thetabar,\etabar}(F_n)\right).
    $$
    \item 
    The rejection region of the test with significance level $\alpha\in(0,1)$ is described by $\sqrt{n}\,T_{\thetabar,\etabar}(F_n)>c_{\alpha,n}$, where 
    \begin{equation}
    \label{eq:critical}
    c_{\alpha,n}=\inf\left\{y: P\left(\left\|\mathcal{G}^b_n\right\|_\infty>y|\widetilde X_1,\ldots,\widetilde X_n\right)\leq \alpha\right\}.
    \end{equation}
    \end{enumerate}
}
We may now prove the relevant asymptotic properties of the bootstrap test.
\begin{proposition}
\label{prop:test-conv-boot}
Assume $F$ is continuous.
	\begin{enumerate}
		\item If $F\in\mathcal{H}_0$, $\lim_{n\to\infty}P(\text{reject }\mathcal{H}_0)\leq \alpha$.
		\item If $F\in\mathcal{H}_1$, $\lim_{n\to\infty}P(\text{reject }\mathcal{H}_0)=1$.
	\end{enumerate}
\end{proposition}

\section{{Simulations}}
\label{sec:simulations}
To give some insights on {the performance of the tests, we provide simulated rejection rates, computed at} significance level $\alpha=0.1$. We consider sample sizes $n=100$ or $n=500$. All simulations reported considered 1000 replications, either for bootstrapping or Monte Carlo simulations. In Subsection~\ref{sec:sim sample mean} we focus on testing dominance between the sample mean and the underlying variable, as this example is particularly simple and relevant. Then, in Subsection~\ref{sec:sim linear comb}, we extend our simulation study to a more general framework, to compare linear combinations with different weights and dimensions. This includes the case in which the weight vectors are not ordered in majorization, for which the theory discussed in Section \ref{sec:dominance-without} does not give any indication.

\subsection{Tests for $\xbar{s} \geq_{st} X$}
\label{sec:sim sample mean}
We first concentrate on testing the null hypothesis $\mathcal H_0:\xbar{s} \geq_{st} X,$ for values of $s=2,3,4.$ We focus on distributions that, for suitable choice of parameters, are known to satisfy this dominance relation, namely, the Pareto, loglogistic or the Fr\'{e}chet families of distributions, with CDFs $\mathcal P(x;sh)=1-\frac1{x^{sh}}$, $x\geq 1$, $\mathcal F(x;sh)=\exp(-x)^{-sh}$, $x\geq 0$ and $\mathcal L(x;sh)=\tfrac{x^{sh}}{1+x^{sh}}$, $x\geq0$, respectively, where $sh$ is a positive shape parameter in all three cases. 
Recall that $\mathcal H_0$ can only hold for distributions with infinite mean, which is the case {of} the families of distributions just mentioned for $sh\leq 1$. 
Differently, when the shape parameter exceeds one, the CDFs of the linear combinations cross, and it should be easier to detect departures from $\mathcal H_0$ as $sh$ gets larger. We are interested in discovering this behaviour by increasing the value of $sh$, as well as the sample size. For all these cases, the null hypothesis is always true, or always false, for every choice of $s$, provided that $sh\leq1$, or $sh>1$, respectively. Then, we analyse a case for which $\mathcal H_0$ is true for some values of $s$ and false for some others. In particular, let $X\stackrel{d}{=}0.45 Y+0.55 Z$, where $Y$ Bernoulli ($\tfrac12$) and $Z$ is a Pareto with $sh=1.$ For this model,  the null hypothesis is false for $s=2,3$ and true for $s=4$ (see Figure~\ref{plot2}).
We also consider the Students' $t$ distribution with $df$ degrees of freedom. In this case, the CDFs of the linear combinations cross for every choice of $df$, except for the case $df=1$, which coincides with the Cauchy distribution, where we expect to observe a power close to $\alpha$. 

As a general comment, not unexpectedly, the bootstrap test shows an overall better power behaviour. However, the running times are considerably smaller for the Cauchy-based test: {when simulating a comparison between some sample mean and the underlying random variable,} while the bootstrap test takes around 15 minutes running on a \texttt{i9} machine, the Cauchy-based test needs just a few seconds to find results when running on a \texttt{i7} machine. 

{The simulated power in the Pareto case is given in Table~\ref{sim.Pareto}. For this model, the CDFs cross for $sh>1$, but they are very close one to another, making it particularly difficult to detect non-dominance, especially for smaller values of $sh$ and $s$ (see Figure \ref{fig:Pareto}). The results show that the bootstrap test outperforms the Cauchy-based test. Moreover, the power increases with respect to $n$, $sh$ and $s$.}
\begin{table}
\centering
\begin{tabular}{r|cc|cc|cc|cc|cc}
 & \multicolumn{2}{c|}{$sh=1$} & \multicolumn{2}{c|}{$sh=2$} & \multicolumn{2}{c|}{$sh=3$} & \multicolumn{2}{c|}{$sh=4$} & \multicolumn{2}{c}{$sh=5$} \\ \hline
$n$ & $100$ & $500$ & $100$ & $500$ & $100$ & $500$ & $100$ & $500$ & $100$ & $500$ \\ \hline
\multicolumn{11}{c}{Bootstrap test} \\ \hline
$s=2$ & 0.000 & 0.000 & 0.051 & 0.104 & 0.133 & 0.358 & 0.208 & 0.585 & 0.271 & 0.709 \\
$s=3$ & 0.001 & 0.000 & 0.078 & 0.153 & 0.189 & 0.559 & 0.307 & 0.808 & 0.395 & 0.924 \\
$s=4$ & 0.000 & 0.000 & 0.067 & 0.208 & 0.206 & 0.684 & 0.329 & 0.889 & 0.433 & 0.965\\ \hline
\multicolumn{11}{c}{Cauchy-based test} \\ \hline
$s=2$ & 0.002 & 0.001 & 0.030 & 0.019 & 0.066 & 0.114 & 0.093 & 0.178 & 0.117 & 0.247 \\
$s=3$ & 0.000 & 0.000 & 0.017 & 0.034 & 0.052 & 0.206 & 0.121 & 0.438 & 0.176 & 0.542 \\
$s=4$ & 0.000 & 0.000 & 0.021 & 0.055 & 0.066 & 0.337 & 0.121 & 0.598 & 0.134 & 0.754 \\ \hline
\end{tabular}
\caption{{Test for $\xbar{s} \geq_{st} X$. }Simulated power with Pareto alternatives.}\label{sim.Pareto}
\end{table}

\begin{figure}
\centering
\begin{tabular}{ccc}
\includegraphics[scale=.375]{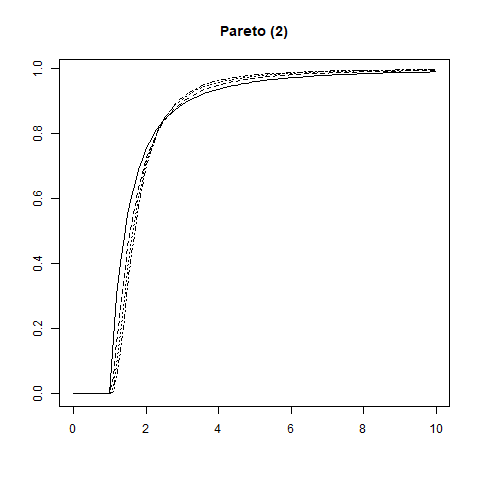} & \quad &
\includegraphics[scale=.375]{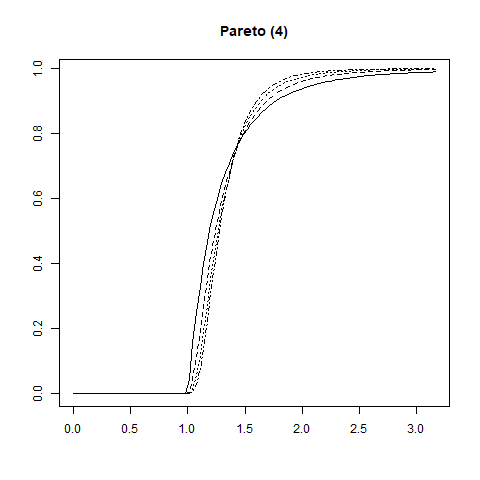}
\end{tabular}
\caption{CDF for the Pareto distribution (solid) and the corresponding CDFs for the sample mean of size $s=2$ (dashed), $s=3$ (dotted) and $s=4$ (dashdotted).}\label{fig:Pareto}
\end{figure}

The performance of the tests in the loglogistic and Fr\'{e}chet cases shows the same trend as in the Pareto case, with respect to {$n$, $s$ and $sh$.} In both these cases, the power behaviour of the test is better than the one observed for the Pareto alternatives, meaning that dominance of sample means for these distributions is easier to identify. The simulated powers are reported in Tables~\ref{sim.LogLogist} and \ref{sim.Frechet}. It is worth noticing that the Cauchy-based test behaves almost as well as the bootstrap test for the loglogistic alternatives, while it lacks power when dealing with the Pareto or Fr\'{e}chet alternatives.
\begin{table}
\centering
\begin{tabular}{r|cc|cc|cc|cc|cc}
 & \multicolumn{2}{c|}{$sh=1$} & \multicolumn{2}{c|}{$sh=2$} & \multicolumn{2}{c|}{$sh=3$} & \multicolumn{2}{c|}{$sh=4$} & \multicolumn{2}{c}{$sh=5$} \\ \hline
$n$ & $100$ & $500$ & $100$ & $500$ & $100$ & $500$ & $100$ & $500$ & $100$ & $500$ \\ \hline
\multicolumn{11}{c}{Bootstrap test} \\ \hline
$s=2$ & 0.002 & 0.000 & 0.109 & 0.206 & 0.304 & 0.678 & 0.451 & 0.888 & 0.557 & 0.956 \\
$s=3$ & 0.000 & 0.000 & 0.130 & 0.323 & 0.393 & 0.862 & 0.586 & 0.985 & 0.739 & 0.999 \\
$s=4$ & 0.000 & 0.000 & 0.145 & 0.423 & 0.402 & 0.941 & 0.637 & 0.998 & 0.767 & 1.000\\ \hline
\multicolumn{11}{c}{Cauchy-based test} \\ \hline
$s=2$ & 0.001 & 0.000 & 0.048 & 0.037 & 0.127 & 0.368 & 0.330 & 0.628 & 0.333 & 0.750 \\
$s=3$ & 0.002 & 0.000 & 0.039 & 0.118 & 0.185 & 0.661 & 0.358 & 0.895 & 0.444 & 0.979 \\
$s=4$ & 0.000 & 0.000 & 0.059 & 0.188 & 0.220 & 0.807 & 0.336 & 0.969 & 0.528 & 0.998 \\ \hline
\end{tabular}
\caption{{Test for $\xbar{s} \geq_{st} X$. }Simulated power with loglogistic alternatives.}\label{sim.LogLogist}
\end{table}
\begin{table}
\centering
\begin{tabular}{r|cc|cc|cc|cc|cc}
 & \multicolumn{2}{c|}{$sh=1$} & \multicolumn{2}{c|}{$sh=2$} & \multicolumn{2}{c|}{$sh=3$} & \multicolumn{2}{c|}{$sh=4$} & \multicolumn{2}{c}{$sh=5$} \\ \hline
$n$ & $100$ & $500$ & $100$ & $500$ & $100$ & $500$ & $100$ & $500$ & $100$ & $500$ \\ \hline
\multicolumn{11}{c}{Bootstrap test} \\ \hline
$s=2$ & 0.001 & 0.000 & 0.087 & 0.156 & 0.212 & 0.507 & 0.322 & 0.768 & 0.403 & 0.856 \\
$s=3$ & 0.001 & 0.000 & 0.097 & 0.254 & 0.279 & 0.748 & 0.432 & 0.940 & 0.558 & 0.987 \\
$s=4$ & 0.002 & 0.000 & 0.114 & 0.321 & 0.296 & 0.865 & 0.478 & 0.975 & 0.612 & 0.996\\ \hline
\multicolumn{11}{c}{Cauchy-based test} \\ \hline
$s=2$ & 0.001 & 0.000 & 0.053 & 0.041 & 0.096 & 0.203 & 0.181 & 0.379 & 0.228 & 0.557 \\
$s=3$ & 0.002 & 0.000 & 0.035 & 0.072 & 0.114 & 0.442 & 0.136 & 0.706 & 0.272 & 0.841 \\
$s=4$ & 0.001 & 0.000 & 0.029 & 0.132 & 0.143 & 0.554 & 0.186 & 0.866 & 0.258 & 0.948 \\ \hline
\end{tabular}
\caption{{Test for $\xbar{s} \geq_{st} X$. }Simulated power with Fr\'{e}chet alternatives.}\label{sim.Frechet}
\end{table}

{The general behaviour is confirmed also in the case of the Student's $t$ distribution, where the power gets larger whenever $df<1$ or $df>1$. As expected, when $df=1$, the simulated power is close to the significance level $\alpha=0.1$ we used for the numerical study. The simulated power values are reported in Table~\ref{sim.Student}.}
\begin{table}
\centering
\begin{tabular}{r|cc|cc|cc|cc|cc}
 & \multicolumn{2}{c|}{$df=.5$} & \multicolumn{2}{c|}{$df=1$} & \multicolumn{2}{c|}{$df=1.5$} & \multicolumn{2}{c|}{$df=3$} & \multicolumn{2}{c}{$df=5$} \\ \hline
$n$ & $100$ & $500$ & $100$ & $500$ & $100$ & $500$ & $100$ & $500$ & $100$ & $500$ \\ \hline
\multicolumn{11}{c}{Bootstrap test} \\ \hline
$s=2$ & 0.444 & 0.903 & 0.118 & 0.117 & 0.247 & 0.483 & 0.691 & 0.984 & 0.883 & 1.000 \\
$s=3$ & 0.513 & 0.953 & 0.122 & 0.102 & 0.290 & 0.696 & 0.805 & 1.000 & 0.954 & 1.000 \\
$s=4$ & 0.540 & 0.974 & 0.108 & 0.111 & 0.287 & 0.726 & 0.866 & 1.000 & 0.986 & 1.000\\ \hline
\multicolumn{11}{c}{Cauchy-based test} \\ \hline
$s=2$ & 0.409 & 0.736 & 0.119 & 0.095 & 0.193 & 0.307 & 0.544 & 0.914 & 0.744 & 0.993 \\
$s=3$ & 0.506 & 0.932 & 0.071 & 0.103 & 0.253 & 0.546 & 0.639 & 0.994 & 0.846 & 1.000 \\
$s=4$ & 0.482 & 0.959 & 0.086 & 0.107 & 0.221 & 0.582 & 0.691 & 1.000 & 0.905 & 1.000 \\ \hline
\end{tabular}
\caption{{Test for $\xbar{s} \geq_{st} X$. }Simulated power with Student's $t$ alternatives.}\label{sim.Student}
\end{table}

{Finally, we consider the mixture case discussed earlier. The results are described in Table~\ref{sim.Mix}. First note that the crossing of the CDF for $s=3$ is by a very slight margin, making it difficult to detect (see Figure \ref{plot2}). Both tests correctly identify the nondominance of $\xbar{2}$, although they need a large sample to reach a reasonable power. For this mixture distribution the Cauchy-based test delivers a larger power than the bootstrap test. }

\begin{figure}[h!]\label{fig:mix}
	\centering
	\includegraphics[scale=0.45,bb=20 35 450 425,clip=true]{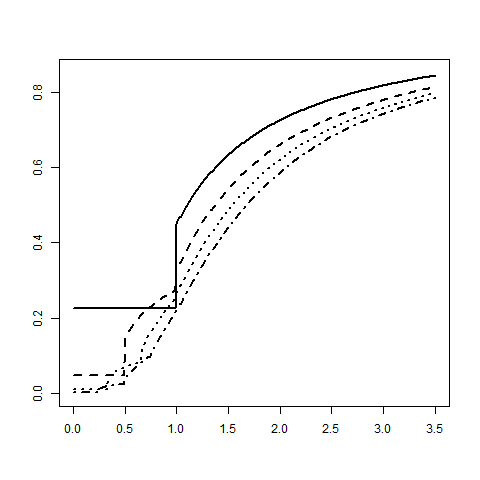}
	\caption{CDFs of $X$ (solid), $\overline{X}_2$ (dashed), $\overline{X}_3$ (dotted), and $\overline{X}_4$ (dashdotted). $X\stackrel{d}{=}0.45 Y+0.55 Z${, $Y$ Bernoulli ($\tfrac12$), $Z$ Pareto with $sh=1$}.}\label{plot2}	
\end{figure}

\begin{table}
\centering
\begin{tabular}{r|cc|cc|cc|cc}
 & \multicolumn{2}{c|}{$s=2$} & \multicolumn{2}{c|}{$s=3$} & \multicolumn{2}{c|}{$s=4$} & \multicolumn{2}{c}{$s=5$} \\ \hline
 $n$ & $100$ & $500$ & $100$ & $500$ & $100$ & $500$ & $100$ & $500$  \\ \hline
Bootstrap test & 0.058 & 0.520 & 0.004 & 0.003 & 0.003 & 0.000 & 0.001 & 0.000 \\
Cauchy-based test & 0.105 & 0.521 & 0.028 & 0.091 & 0.006 & 0.002 & 0.002 & 0.000 \\ \hline
\end{tabular}
\caption{Test for $\xbar{s} \geq_{st} X$. Simulated power with alternative $X\stackrel{d}{=}0.45 Y+0.55 Z${, $Y$ Bernoulli ($\tfrac12$), $Z$ Pareto with $sh=1$}.}\label{sim.Mix}
\end{table}

\subsection{Testing dominance between different linear combinations}
\label{sec:sim linear comb}
{In this subsection, we compare linear combinations with different weight vectors. The sample size in this case is fixed, $n=500$. We start by testing $\mathcal H_0:\theta_1X_1+\theta_2X_2\geq_{st}X$.} The results are reported in Tables~\ref{sim.Weighted.boot} and \ref{sim.Weighted.Cauchy}, using the same set of distributions as alternatives {as in Subsection~\ref{sec:sim sample mean}}.
{The results exhibit the same pattern as in the comparison of sample means}. {Moreover, power is clearly higher for weight vectors close to $(\tfrac 12,\tfrac12)$, i.e. those defining the sample mean, and lower when the weights concentrate on a single observation. This indicates that majorization influences non-dominance (CDF crossings), in line with the classical finite-mean theory linking majorization of weights to convex-order comparisons of weighted sums \citep[Chap.~3.A]{shaked}}. Expectedly, although the Cauchy-based test is much faster to run, it has a significant{ly} lower power.
\begin{table}
\centering
\begin{tabular}{r|ccccc}
Pareto ($sh$) & $sh=1$ & $sh=2$ & $sh=3$ & $sh=4$ & $sh=5$ \\ \hline
$(0.4,\,0.6)$ & 0.000 & 0.100 & 0.372 & 0.570 & 0.708 \\
$(0.2,\,0.8)$ & 0.000 & 0.119 & 0.344 & 0.484 & 0.618 \\ \hline
loglogistic ($sh$) & $sh=1$ & $sh=2$ & $sh=3$ & $sh=4$ & $sh=5$ \\ \hline
$(0.4,\,0.6)$ & 0.000 & 0.220 & 0.670 & 0.885 & 0.967 \\
$(0.2,\,0.8)$ & 0.000 & 0.197 & 0.587 & 0.780 & 0.859 \\ \hline
Fr\'{e}chet ($sh$) & $sh=1$ & $sh=2$ & $sh=3$ & $sh=4$ & $sh=5$ \\ \hline
$(0.4,\,0.6)$ & 0.000 & 0.162 & 0.514 & 0.761 & 0.869 \\
$(0.2,\,0.8)$ & 0.000 & 0.157 & 0.451 & 0.646 & 0.765 \\ \hline
Student ($df$) & $df=0.5$ & $df=1$ & $df=1.5$ & $df=3$ & $df=5$ \\ \hline
$(0.4,\,0.6)$ & 0.899 & 0.121 & 0.490 & 0.984 & 1.000 \\
$(0.2,\,0.8)$ & 0.850 & 0.119 & 0.435 & 0.940 & 0.986 \\ \hline
\end{tabular}
\caption{Simulated power of the bootstrap test for weighted combinations.}
\label{sim.Weighted.boot}
\end{table}

\begin{table}
\centering
\begin{tabular}{r|ccccc}
Pareto ($sh$) & $sh=1$ & $sh=2$ & $sh=3$ & $sh=4$ & $sh=5$ \\ \hline
$(0.4,\,0.6)$ & 0.000 & 0.020 & 0.042 & 0.136 & 0.244 \\
$(0.2,\,0.8)$ & 0.000 & 0.006 & 0.069 & 0.054 & 0.104 \\ \hline
loglogistic ($sh$) & $sh=1$ & $sh=2$ & $sh=3$ & $sh=4$ & $sh=5$ \\ \hline
$(0.4,\,0.6)$ & 0.000 & 0.056 & 0.348 & 0.594 & 0.740 \\
$(0.2,\,0.8)$ & 0.000 & 0.025 & 0.160 & 0.363 & 0.445 \\ \hline
Fr\'{e}chet ($sh$) & $sh=1$ & $sh=2$ & $sh=3$ & $sh=4$ & $sh=5$ \\ \hline
$(0.4,\,0.6)$ & 0.000 & 0.035 & 0.147 & 0.353 & 0.514 \\
$(0.2,\,0.8)$ & 0.000 & 0.022 & 0.107 & 0.161 & 0.264 \\ \hline
Student ($df$) & $df=0.5$ & $df=1$ & $df=1.5$ & $df=3$ & $df=5$ \\ \hline
$(0.4,\,0.6)$ & 0.782 & 0.104 & 0.303 & 0.917 & 0.990 \\
$(0.2,\,0.8)$ & 0.661 & 0.102 & 0.233 & 0.637 & 0.836 \\ \hline
\end{tabular}
\caption{Simulated power of the Cauchy-based test for weighted combinations.}
\label{sim.Weighted.Cauchy}
\end{table} 

{We now focus on the tests for $\mathcal H_0:\theta_1X_1+\theta_2X_2\geq_{st}\xbar{2}$. The results are reported in Tables~\ref{sim.Weighted.boot1} and \ref{sim.Weighted.Cauchy1}. Even in this case, the bootstrap test consistently outperforms the Cauchy-based test.}
\begin{table}
\centering
\begin{tabular}{r|ccccc}
Pareto ($sh$) & $sh=1$ & $sh=2$ & $sh=3$ & $sh=4$ & $sh=5$ \\ \hline
$(0.1,\,0.9)$ vs. $(0.5,\,0.5)$ & 0.000 & 0.063 & 0.159 & 0.258 & 0.336 \\
$(0.25,\,0.75)$ vs. $(0.5,\,0.5)$ & 0.037 & 0.048 & 0.084 & 0.118 & 0.138 \\ \hline
loglogistic ($sh$) & $sh=1$ & $sh=2$ & $sh=3$ & $sh=4$ & $sh=5$ \\ \hline
$(0.1,\,0.9)$ vs. $(0.5,\,0.5)$ & 0.004 & 0.137 & 0.356 & 0.530 & 0.663 \\
$(0.25,\,0.75)$ vs. $(0.5,\,0.5)$ & 0.039 & 0.103 & 0.165 & 0.243 & 0.301 \\ \hline
Fr\'{e}chet ($sh$) & $sh=1$ & $sh=2$ & $sh=3$ & $sh=4$ & $sh=5$ \\ \hline
$(0.1,\,0.9)$ vs. $(0.5,\,0.5)$ & 0.003 & 0.105 & 0.276 & 0.409 & 0.523 \\
$(0.25,\,0.75)$ vs. $(0.5,\,0.5)$ & 0.044 & 0.090 & 0.142 & 0.185 & 0.222 \\ \hline
Student ($df$) & $df=0.5$ & $df=1$ & $df=1.5$ & $df=3$ & $df=5$ \\ \hline
$(0.1,\,0.9)$ vs. $(0.5,\,0.5)$ & 0.351 & 0.115 & 0.239 & 0.724 & 0.898\\
$(0.25,\,0.75)$ vs. $(0.5,\,0.5)$ & 0.210 & 0.121 & 0.150 & 0.289 & 0.428 \\ \hline
\end{tabular}
\caption{Simulated power of the bootstrap test between different weighted combinations.}
\label{sim.Weighted.boot1}
\end{table} 

\begin{table}
\centering
\begin{tabular}{r|ccccc}
Pareto ($sh$) & $sh=1$ & $sh=2$ & $sh=3$ & $sh=4$ & $sh=5$ \\ \hline
$(0.1,\,0.9)$ vs. $(0.5,\,0.5)$ & 0.002 & 0.011 & 0.032 & 0.064 & 0.120 \\
$(0.25,\,0.75)$ vs. $(0.5,\,0.5)$ & 0.028 & 0.024 & 0.034 & 0.023 & 0.022 \\ \hline
loglogistic ($sh$) & $sh=1$ & $sh=2$ & $sh=3$ & $sh=4$ & $sh=5$ \\ \hline
$(0.1,\,0.9)$ vs. $(0.5,\,0.5)$ & 0.001 & 0.024 & 0.100 & 0.239 & 0.301 \\
$(0.25,\,0.75)$ vs. $(0.5,\,0.5)$ & 0.045 & 0.053 & 0.052 & 0.089 & 0.101 \\ \hline
Fr\'{e}chet ($sh$) & $sh=1$ & $sh=2$ & $sh=3$ & $sh=4$ & $sh=5$ \\ \hline
$(0.1,\,0.9)$ vs. $(0.5,\,0.5)$ & 0.002 & 0.015 & 0.092 & 0.129 & 0.161 \\
$(0.25,\,0.75)$ vs. $(0.5,\,0.5)$ & 0.036 & 0.033 & 0.020 & 0.046 & 0.044 \\ \hline
Student ($df$) & $df=0.5$ & $df=1$ & $df=1.5$ & $df=3$ & $df=5$ \\ \hline
$(0.1,\,0.9)$ vs. $(0.5,\,0.5)$ & 0.156 & 0.084 & 0.160 & 0.481 & 0.670 \\
$(0.25,\,0.75)$ vs. $(0.5,\,0.5)$ & 0.114 & 0.101 & 0.113 & 0.141 & 0.180 \\ \hline
\end{tabular}
\caption{Simulated power of the Cauchy-based test between different weighted combinations.}
\label{sim.Weighted.Cauchy1}
\end{table} 

{We finally apply our tests to the case when the vectors are not ordered by majorization. In this case, the theory does not provide any indication. Dominance can still hold, but we do not know in which direction, that is, we don't know which linear combination dominates the other. With this motivation, we test $\mathcal H_0:\theta_1X_1+\theta_2X_2+\theta_3X_3\geq_{st}\eta_1X_1+\eta_2X_2+\eta_3X_3$ as well the reversed null hypothesis $\mathcal H_0:\eta_1X_1+\eta_2X_2+\eta_3X_3\geq_{st}\theta_1X_1+\theta_2X_2+\theta_3X_3$. In all such cases $X$ has a loglogistic distribution. Two choices of $\overline{\eta}$ and $\overline{\eta}$ such that $\overline{\eta}\not\succ \overline{\theta}$ and $\overline{\theta}\not\succ \overline{\eta}$ are considered in Table~\ref{sim.direction1} and Table~\ref{sim.direction2}. In the first case, the two CDFs are very close and almost indistinguishable, similarly to the Cauchy case. The tests behave accordingly, showing power close to $\alpha$ in one direction and slightly larger than $\alpha$ in the other direction (where the null is false). In the latter case, one distribution dominates the other more clearly, especially for smaller values of the shape parameter. Both tests are able to detect this situation, showing power close to 0 when the null is true, and always larger than 0.7 when the null is false. We finally study a three-dimensional case where $\overline{\eta}$ and $\overline{\theta}$ are ordered in majorization. The results, reported in Table~\ref{sim.direction}, show that, in all cases considered, the power is close to 0 when the null is true, and close to 1 when it is false.}

\begin{table}
\centering
\begin{tabular}{r|ccc}
 & $sh=1$ & $sh=0.8$ & $sh=0.6$ \\ \hline
\multicolumn{4}{c}{Bootstrap test}\\ \hline
$\thetabar=(0.1,\,0.35,\,0.55)$ vs. $\etabar=(0.15,\,0.25,\,0.6)$ & 0.202 & 0.209 & 0.236 \\
$\etabar=(0.15,\,0.25,\,0.6)$ vs. $\thetabar=(0.1,\,0.35,\,0.55)$ & 0.109 & 0.104 & 0.115 \\ \hline
\multicolumn{4}{c}{Cauchy-based test}\\ \hline
$\thetabar=(0.1,\,0.35,\,0.55)$ vs. $\etabar=(0.15,\,0.25,\,0.6)$ & 0.120 & 0.126 & 0.166 \\
$\etabar=(0.15,\,0.25,\,0.6)$ vs. $\thetabar=(0.1,\,0.35,\,0.55)$ & 0.093 & 0.069 & 0.089 \\ \hline
\end{tabular}
\caption{Simulated power to detect the direction of stochastic dominance in the presence of nonmajorized weights (underlying distribution is the loglogistic with shape parameter $sh$).}
\label{sim.direction1}
\end{table}
\begin{table}
\centering
\begin{tabular}{r|ccc}
 & $sh=1$ & $sh=0.8$ & $sh=0.6$ \\ \hline
\multicolumn{4}{c}{Bootstrap test}\\ \hline
$\thetabar=(0.09,\,0.41,\,0.5)$ vs. $\etabar=(0.1,\,0.1,\,0.8)$ & 0.05 & 0.002 & 0.001 \\
$\etabar=(0.1,\,0.1,\,0.8)$ vs. $\thetabar=(0.09,\,0.41,\,0.5)$ & 0.941 & 0.941 & 0.864 \\ \hline
\multicolumn{4}{c}{Cauchy-based test}\\ \hline
$\thetabar=(0.09,\,0.41,\,0.5)$ vs. $\etabar=(0.1,\,0.1,\,0.8)$  & 0.003 & 0.006 & 0.005 \\
$\etabar=(0.1,\,0.1,\,0.8)$ vs. $\thetabar=(0.09,\,0.41,\,0.5)$ & 0.887 & 0.795 & 0.735 \\ \hline
\end{tabular}
\caption{Simulated power to detect the direction of stochastic dominance in the presence of nonmajorized weights (underlying distribution is the loglogistic with shape parameter $sh$).}
\label{sim.direction2}
\end{table}

\begin{table}
	\centering
	\begin{tabular}{r|ccc}
		& $sh=1$ & $sh=0.8$ & $sh=0.6$ \\ \hline
		\multicolumn{4}{c}{Bootstrap test}\\ \hline
{$\thetabar=(0.2,\,0.3,\,0.5)$ vs. $\etabar=(0.1,\,0.1,\,0.8)$}& 0.998 & 0.987 & 0.972 \\
{$\etabar=(0.1,\,0.1,\,0.8)$ vs. $\thetabar=(0.2,\,0.3,\,0.5)$} & 0.003 & 0.001 & 0.000 \\ \hline
		\multicolumn{4}{c}{Cauchy-based test}\\ \hline
	{$\thetabar=(0.2,\,0.3,\,0.5)$ vs. $\etabar=(0.1,\,0.1,\,0.8)$} & 0.985 & 0.956 & 0.921 \\
	{$\etabar=(0.1,\,0.1,\,0.8)$ vs. $\thetabar=(0.2,\,0.3,\,0.5)$}& 0.001 & 0.000 & 0.001 \\ \hline
	\end{tabular}
	\caption{Simulated power {with majorized weights} (underlying distribution is the loglogistic with shape parameter $sh$).}
	\label{sim.direction}
\end{table} 

\section{Empirical analysis}\label{sec:application}

We now apply the testing procedures discussed to assess whether pooling or some balancing in insurance of catastrophic events has some effect in comparing underlying risks, namely, identifying some stochastic dominance among the different insurance choices.

Let $X$ denote the severity of a single loss. We compare $X$ with the uniform aggregation $\xbar{s}$.
In an insurance-pooling interpretation, this represents the normalized loss per exposure in an idealized homogeneous pool of $s$ independent risks.
The stochastic dominance relation $\xbar{s}\ge_{st} X$ therefore means that pooling $s$ risks leads to a normalized loss which is stochastically larger than the loss of a single risk. This is the opposite of the usual diversification	intuition for finite-mean risks.

The data sets are taken from the \texttt{CASdatasets} package, a curated collection of actuarial and insurance datasets originally developed for \emph{Computational	Actuarial Science with R} (see \cite{DutangCharpentier2026CASdatasets} and
\cite{Charpentier2014CAS}). 

We consider two datasets with different risk interpretations. The first one is
the French motor third-party liability claim-severity dataset \texttt{freMTPL2sev}, with $n=26~444$ claim amounts and the corresponding policy identifiers. In our analysis, \(X\) denotes a positive claim amount, measured by the variable \texttt{ClaimAmount}. 
The second dataset that we study is the French fire claim-severity dataset
\texttt{freclaimset3fire9207}. The data contain $n=22~721$ individual fire claim payments from French insurance portfolios. In this case, \(X\) denotes the severity of a positive fire claim, measured by the variable \texttt{paid\_Y10}.
With respect to the theoretical condition implying \eqref{SD-cx}, namely the concavity of the inverted distribution \( 1-F(\tfrac1x)\), we use the empirical version \(1-\mathbb F_n(\tfrac1x)\) as a graphical diagnostic.
This is only a visual check and does not amount to a formal test of the condition. As shown in Figure~\ref{fig:conc}, the French motor third-party liability data visibly depart from a concave shape, also because of point masses, while the French fire claim-severity data appear more compatible with concavity.

\begin{figure}
\centering
\includegraphics[scale=.35]{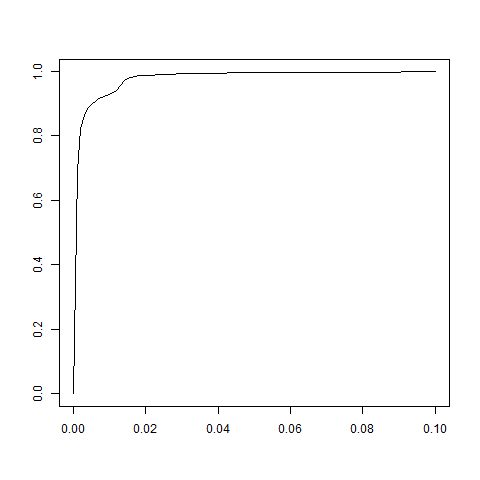}
\hspace*{.2cm}
\includegraphics[scale=.35]{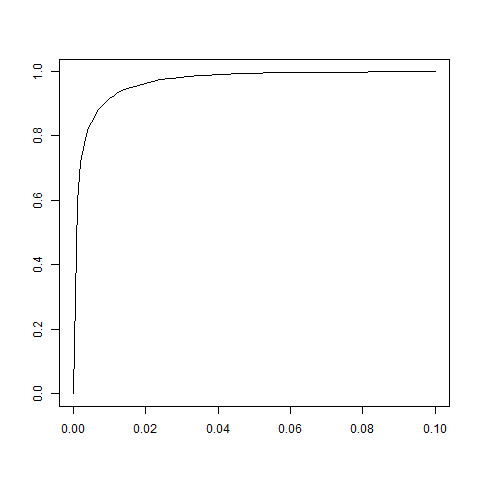}
\caption{Empirical CDF of the inverted distribution for the French motor third-party liability (left), and French fire claim-severity (right).}
\label{fig:conc}
\end{figure}

For both these claim-severity data, the pooling interpretation is direct:
\(X\) is the severity of an individual claim and \(\xbar{s}\) is the average claim
severity in an idealized homogeneous pool of \(s\) independent claims with the same severity distribution. 
We consider \(s=2,5,10,20,50,100\), allowing to assess how the dominance relation
changes as the degree of pooling increases.
%

We first focus on the motor dataset.
As noted above, the theoretical framework of \cite{Chen2025Combinations} seems not applicable in this case. 
However, a dominance relation between $\xbar{s}$ and $X$ might hold, or not, depending also on the values of $s$. In such a situation, the dominance tests we are proposing are useful for assessing the empirical evidence against the dominance and for quantifying the size of its violations. 
Indeed, \(p\)-values indicate whether these deviations are large relative to the chosen calibration.
The results show that, for small values of \(s\), the violations of the
dominance are statistically significant, with \(p\)-values essentially equal to zero for \(s=2\) and \(s=5\). 
This is also visible from the
plots of \(\mathbb F_n\) against the empirical distribution function of
\(\xbar{s}\), denoted by \(\mathbb F_{n,s}\), where the two curves exhibit clear
crossings. However, as the aggregation size increases to $s=50$, these violations become
much smaller, both graphically and statistically. Thus, in this dataset, pooling
a large number of claim severities moves the normalized pooled loss
towards the unexpected-dominance regime. In this sense, the usual diversification
intuition may fail. Nevertheless, for $s=100$ the behaviour changes again, as now the distributions exhibit a crossing in the right tail.

%
For the French fire data, the plot of \(1-\mathbb F_n(\tfrac1x)\)
(see Figure~\ref{fig:conc}) appears concave at the scale displayed in the
figure, although a closer inspection near the origin reveals local deviations
from concavity. The losses are also much larger in magnitude, pointing
to a substantially heavier-tailed severity distribution. Nevertheless, the
dominance behaviour is not immediate from this diagnostic, and the proposed
tests provide useful additional information.

The results differ substantially from those obtained for the motor data. For small aggregation sizes, the dominance null is not rejected, and the corresponding \(p\)-values are large. 
As \(s\) increases, the tests become borderline or reject the null. This behaviour is also visible in Figure~\ref{fig:fire}, where
the empirical distributions exhibit increasingly clear crossings as \(s\) grows.
Thus, for the fire data, the departures from dominance are small for small pools and large for larger pools. This is
consistent with a finite-mean behaviour: as the degree of pooling increases, the distribution of \(\xbar{s}\) starts to
concentrate around the mean, leading to crossings. In this sense, the fire data appear closer to the
classic diversification notion than the motor data, although the
heavy-tailed nature of the losses makes this behaviour visible only for large aggregation sizes.

To conclude, in both datasets, for some aggregation sizes, the normalized pooled loss
\(\bar X_s\) appears to become stochastically larger than the individual loss
\(X\), as indicated by our tests. In this sense, the proposed procedures are
useful for detecting how the effect of pooling changes with the aggregation
level.
The empirical analysis also shows that the tests are particularly informative for
real datasets. In simulation studies, the distributions usually have a regular shape, and the effect of \(s\) is often governed by a small number of
shape parameters. By contrast, the insurance datasets considered here exhibit
more irregular empirical shapes, making it harder to predict the effect
of aggregation.

\begin{figure}
	\centering
	\begin{tabular}{ccc}
		\includegraphics[scale=1.1]{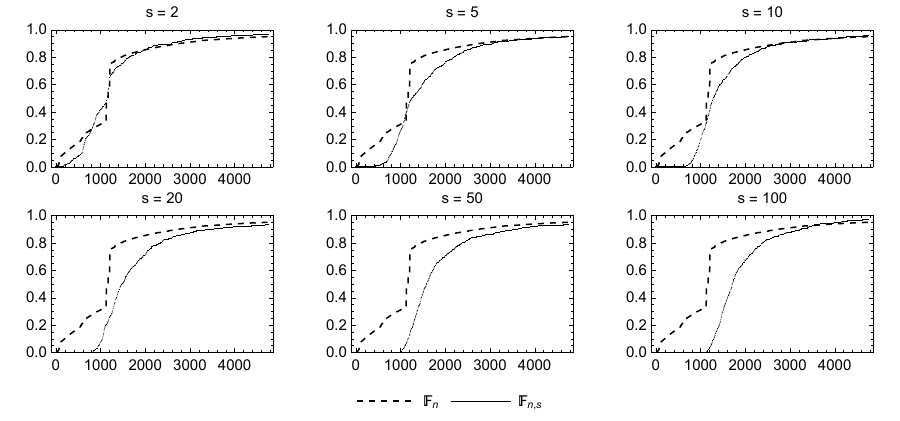}
	\end{tabular}
	\caption{French motor claim-severity data. Comparison between $\F_n$ and the CDF of $\bar X_s$, denoted with $\F_{n,s},$ for different values of $s$.}\label{fig:motor}
\end{figure}


\begin{table}[ht]
	\centering
	
	\begin{tabular}{rrrrrr}
		\hline
		\(s\) 
		& \(\sqrt n\,T_{(s)}\) 
		& \(\sqrt n\,c^{\mathcal C}_{0.90}\) 
		& $p({\mathcal C})$
		& \(\sqrt n\,c^{Boot}_{0.90}\) 
		& $p({Boot})$ \\
		\hline
		2   & 17.073 & 1.136 & 0.000 & 1.321 & 0.000 \\
		5   & 14.082 & 1.528 & 0.000 & 1.946 & 0.000 \\
		10  & 0.521  & 2.101 & 0.842 & 2.216 & 0.932 \\
		20  & 0.789  & 2.868 & 0.744 & 2.743 & 0.804 \\
		50  & 1.941  & 4.399 & 0.520 & 4.528 & 0.495 \\
		100 & 3.128  & 5.285 & 0.400 & 6.091 & 0.421 \\
		200 & 3.412  & 8.704 & 0.520 & 8.485 & 0.441 \\
		\hline
	\end{tabular}
	\caption{Cauchy-based and bootstrap tests for the French motor claim-severity data.}
	\label{tab:french-motor-cauchy-bootstrap-scaled}
\end{table}

\begin{figure}
	\centering
	\begin{tabular}{ccc}
		\includegraphics[scale=1.1]{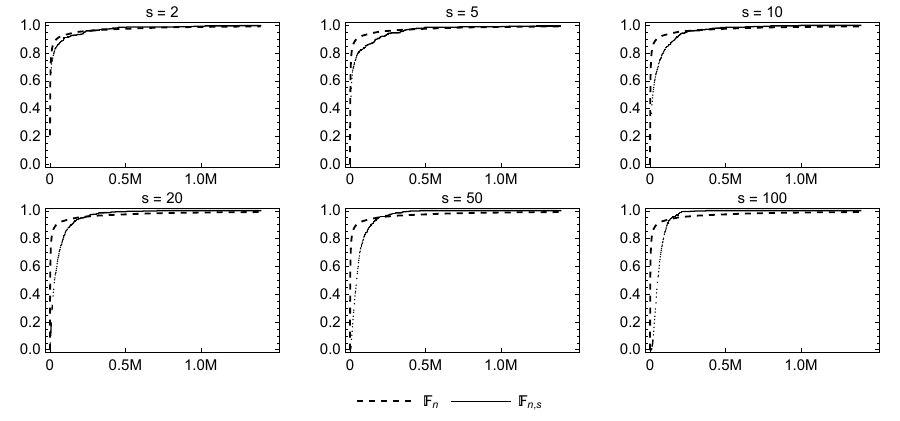}
	\end{tabular}
	\caption{French fire claim-severity data. Comparison between $\F_n$ and the CDF of $\bar X_s$, denoted with $\F_{n,s},$ for different values of $s$.}\label{fig:fire}
\end{figure}

\begin{table}[ht]
	\centering
	\begin{tabular}{rrrrrr}
		\hline
		\(s\) 
		& \(\sqrt n\,T_{(s)}\) 
		& \(\sqrt n\,c^{\mathcal C}_{0.90}\) 
		& \(p(\mathcal C)\) 
		& \(\sqrt n\,c^{Boot}_{0.90}\) 
		& \(p(Boot)\) \\
		\hline
		2   & 0.354 & 1.087 & 0.880 & 1.446 & 0.998 \\
		5   & 0.824 & 1.461 & 0.494 & 1.332 & 0.373 \\
		10  & 1.470 & 2.162 & 0.296 & 2.308 & 0.341 \\
		20  & 2.556 & 2.916 & 0.156 & 2.666 & 0.124 \\
		50  & 4.389 & 4.435 & 0.106 & 3.594 & 0.050 \\
		100 & 5.701 & 5.829 & 0.104 & 5.579 & 0.100 \\
		\hline
	\end{tabular}
		\caption{Cauchy-based and bootstrap tests for the French fire claim-severity data.}
	\label{tab:french-fire-cauchy-bootstrap-scaled}
\end{table}

\begin{appendix}

\section{Proofs}
\label{ap:WeigConvH}

\begin{proof}[Proof of Proposition~\ref{t1}]
It is well known that $U\geq_{st}V$ implies $\sum_k \pi_k U_k\geq_{st}\sum_k \pi_k V_k$, by closure of stochastic dominance under convolution and nonnegative scaling (see for instance Theorem 1.A.3 of \cite{shaked}). 
Taking $U=\sum_i\theta_iX_i$ and $V=\sum_j\eta_jX_j$, the first assertion follows immediately since 
$$
\sum_{\ell=1}^{k} \pi_\ell U_\ell=\sum_i^{nk} (\overline{\theta}\otimes \overline{\pi})_i X_i \quad \text{and} \quad \sum_{\ell=1}^{k} \pi_\ell V_\ell=\sum_j^{mk} (\overline{\eta}\otimes \overline{\pi})_jX_j, 
$$
where {$U_\ell$ and $V_\ell$} are \iid\ copies of $U$ and $V$, respectively.

The second assertion is an immediate consequence of assertion 1.

For the third assertion, {it is enough to prove the result for} a single $s$-split. Let $\etabar=(\eta_1,\ldots,\eta_k)$ and fix $j\in\{1,\ldots,k\}$ so that
$$
\thetabar=(\eta_1,\ldots,\eta_{j-1},\underbrace{\tfrac{\eta_j}{s},\ldots,\tfrac{\eta_j}{s}}_{s\ \text{times}},\eta_{j+1},\ldots,\eta_k).
$$
Write
\[
\sum_{i=1}^{k}\eta_i X_i = \eta_j X_j + \sum_{i\ne j}\eta_i X_i =: \eta_j X_j + Z,
\]
where $Z$ is independent of $(X_j,X_{k+1},\dots,X_{k+s-1})$. After splitting $\eta_j$ into $s$ equal parts we obtain
$$
\sum_{i=1}^{k+s}\theta_i X_i = \frac{\eta_j}{s}\left(X_j+X_{k+1}+\cdots+X_{k+s-1}\right) + Z
= \eta_j\overline{X}_s^{(j)} + Z,
$$
where $\overline{X}_s^{(j)}:=\frac1s(X_j+X_{k+1}+\cdots+X_{k+s-1})$ has the same distribution as $\xbar{s}$.
Multiplying by the nonnegative constant $\eta_j$ preserves the dominance $\overline{X}_s^{(j)}\geq_{st} X_j$, hence $\eta_j\,\overline{X}_s^{(j)}\geq_{st} \eta_j X_j$.
Finally, adding $Z$ to both sides preserves the stochastic order, so
$$
\sum_{i=1}^{k+s}\theta_i X_i = \eta_j\overline{X}_s^{(j)} + Z\geq_{st}\eta_j X_j + Z = \sum_{j=1}^k\eta_j X_j,
$$
completing the proof of this assertion.

Concerning assertion 4., write $\xbar{s+k}$ as the weighted average:
$$
\xbar{s+k}
=\frac{s}{s+k}\left(\frac1s\sum_{i=1}^s X_i\right)
+\frac{k}{s+k}\left(\frac1k\sum_{i=s+1}^{s+k}X_i\right)
=\frac{s}{s+k}\,S_1+\frac{k}{s+k}\,S_2,
$$
where $S_1\stackrel{d}{=}\xbar{s}$, $S_2\stackrel{d}{=}\xbar{k}$ and are independent. Replacing $\xbar{s}$ with $\xbar{t}$, which is stochastically smaller by assumption, gives the result.
\end{proof}

\begin{proof}[Proof of Proposition~\ref{prop:Hadamard}]
The weighted convolution operator is obviously linear with respect to each CDF argument (see (\ref{eq:wei-conv-op})).
Therefore, given $t\in\mathbb{R}$ and $h_t,\,h\in \ell^\infty(\mathbb{R})$, such that $\lim_{t\rightarrow0}\|h_t-h\|_\infty=0$, using this bilinearity,
$$
L_{\theta_1,\theta_2,F+th_t}(F+th_t)-L_{\theta_1,\theta_2,F}(F)=t\left(L_{\theta_1,\theta_2,F}(h_t)+L_{\theta_1,\theta_2,h_t}(F)\right)+t^2L_{\theta_1,\theta_2,h_t}(h_t),
$$
it follows that, when $\thetabar=(\theta_1,\theta_2)$, the Hadarmard derivative at $F$ in the direction $h$ of the operator is $L_{(\theta_1,\theta_2)}'(F;h)=L_{\theta_1,\theta_2,F}(h)+L_{\theta_1,\theta_2,h}(F)$. For a general vector $\thetabar=(\theta_1,\ldots,\theta_s)$ note first that $L_{(\theta_1,\ldots,\theta_s)}(F)=L_{1,\theta_s,F}(L_{(\theta_1,\ldots,\theta_{s-1})}(F))$.
Applying the chain rule to this representation, one gets
$$
L_{(\theta_1,\ldots,\theta_s)}'(F;h)= L_{1,\theta_s,F}(L_{(\theta_1,\ldots,\theta_{s-1})}'(F;h)) + L_{1,\theta_s,h}(L_{(\theta_1,\ldots,\theta_{s-1})}(F)).
$$
Iterating this expression concludes the proof of the representation (\ref{eq:hadamard}).
\end{proof}

\label{ap:uConv}

\begin{proof}[Proof of Proposition~\ref{prop:unif-conv}]
	It is enough to prove the claim for $\bar\theta=(\theta_1,\theta_2)$, since the general case follows by iterating the same argument. 
	Write
	\begin{equation}
		\F_{n,\bar\theta}(x)-F_{\bar\theta}(x)=
		\int \F_n\!\left(\frac{x-\theta_2 t}{\theta_1}\right)\,d\F_n(t)
		-
		\int F\!\left(\frac{x-\theta_2 t}{\theta_1}\right)\,dF(t)=A_n(x)+B_n(x),
		\label{eq:decomp_unif}
	\end{equation}
	where 
		\begin{eqnarray*}
	A_n(x) & :=	& \int \F_n\!\left(\frac{x-\theta_2 t}{\theta_1}\right)-F\!\left(\frac{x-\theta_2 t}{\theta_1}\right)\,dF(t),\\[1.5ex]
	B_n(x) & := & \int \F_n\!\left(\frac{x-\theta_2 t}{\theta_1}\right)\,d(\F_n-F)(t).
	\end{eqnarray*}
	
	We prove that $\sup_{x\in\R}|A_n(x)|\to_{a.s.} 0$ and $\sup_{x\in\R}|B_n(x)|\to_{a.s.} 0$.
	By the triangle inequality and the Glivenko-Cantelli theorem,
	\begin{align*}
		\sup_{x\in\R}|A_n(x)|
		&\leq
		\sup_{x\in\R}\int 
		\left|\F_n\!\left(\frac{x-\theta_2 t}{\theta_1}\right)-F\!\left(\frac{x-\theta_2 t}{\theta_1}\right)\right|\,dF(t)\\[1.5ex]
		&\leq
		\int \sup_{x\in\R}
		\left|\F_n\!\left(\frac{x-\theta_2 t}{\theta_1}\right)-F\!\left(\frac{x-\theta_2 t}{\theta_1}\right)\right|\,dF(t)\\[1.5ex]
		&=
		\sup_{y\in\R}\left|\F_n(y)-F(y)\right|\to_{a.s.}0.
	\end{align*}
Denote $F_k(x)=F(\tfrac xk)$. Then write
	$$
	\int \F_n\!\left(\frac{x-\theta_2 t}{\theta_1}\right)\,dF(t)
	=
	\int \F_n\!\left(\frac{x}{\theta_1}-u\right)\,dF_{\tfrac{\theta_2}{\theta_1}}(u)
	=
	\mathbb{F}_n\ast F_{\tfrac{\theta_2}{\theta_1}}\!\left(\frac{x}{\theta_1}\right).
	$$
	Proceeding similarly for the remaining term in $B_n(x)$ and using commutativity of convolution, we obtain
	$$
	B_n(x)
	=
	\int \mathbb{F}_{\tfrac{\theta_2}{\theta_1},n}\!\left(\frac{x}{\theta_1}-u\right)-F_{\tfrac{\theta_2}{\theta_1}}\!\left(\frac{x}{\theta_1}-u\right)\,d\mathbb{F}_n(u),
	$$
	where $\F_{\tfrac{\theta_2}{\theta_1},n}$ denotes the empirical CDF of the sample $\tfrac{\theta_2}{\theta_1}X_1, ...,\tfrac{\theta_2}{\theta_1}X_n$. Since $\F_n$ is a CDF, 
	$$
	|B_n(x)|
	\le
	\sup_{u\in\mathbb R}\left|\mathbb{F}_{\tfrac{\theta_2}{\theta_1},n}\!\left(\frac{x}{\theta_1}-u\right)-F_{\tfrac{\theta_2}{\theta_1}}\!\left(\frac{x}{\theta_1}-u\right)\right|
	=
	\sup_{y\in\mathbb R}\left|\mathbb{F}_{\tfrac{\theta_2}{\theta_1},n}(y)-F_{\tfrac{\theta_2}{\theta_1}}(y)\right|.
	$$
	Taking the supremum over $\R$ yields
	$$
	\sup_{x\in\mathbb R}|B_n(x)|
	\le
	\sup_{y\in\mathbb R}\left|\mathbb{F}_{\tfrac{\theta_2}{\theta_1},n}(y)-F_{\tfrac{\theta_2}{\theta_1}}(y)\right|
	\to_{a.s.}0,
	$$
	again by the Glivenko-Cantelli theorem.
	
	\smallskip
	We conclude that 
    $\sup_{x\in\mathbb R}\left|\mathbb{F}_{n,\bar\theta}(x)-\mathbb{F}_{\bar\theta}(x)\right|
	\to_{a.s.}0,$ which proves Proposition~\ref{prop:unif-conv} for $\bar\theta=(\theta_1,\theta_2)$.
\end{proof}

\label{ap:emp-proc}
\begin{proof}[Proof of Theorem~\ref{thm:weak-Fthetabar}]
{The continuity of $F$ implies that $\sqrt n\,(\F_n-F)\rightsquigarrow\mathbb{B}_F$ in $\ell^\infty(\R)$. Then, t}he convergence of the empirical process is an immediate consequence of the functional delta method as
$\sqrt{n}\,\left(\mathbb{F}_{n,\thetabar}-F_{\thetabar}\right)=\sqrt{n}\,\left(L_{\thetabar}(\mathbb{F}_n)-L_{\thetabar}(F)\right)
\rightsquigarrow L_{\thetabar}'(F;\mathbb{B}_F)=:\mathbb{H}_{\thetabar}^F$.

To prove the representation of $\mathbb{H}_{\thetabar}^F$, first, note that, using (\ref{eq:hadamard}), we know that 
\begin{equation}
\label{eq:covaT}
L_{\thetabar}'(F;\BB)=\sum_{(a_1,\ldots,a_s)\in D} L_{1,\theta_s,a_s}\circ L_{1,\theta_{s-1},a_{s-1}}\circ\cdots\circ L_{1,\theta_3,a_3}\circ L_{\theta_1,\theta_2,a_1},
\end{equation}
where $D=\{(\BB,F,\ldots,F),\,(F,\BB,F,\ldots,F),\ldots,(F,\ldots,F,\BB)\}$, as follows from Proposition~\ref{prop:Hadamard}. 
Consider some fixed $(a_1,\ldots,a_s)\in D$, that is, the set of vectors that have $\BB$ in its $j$-th coordinate and $F$ in the others. Then,
\begin{eqnarray*}
\lefteqn{L_{1,\theta_s,a_s}\circ L_{1,\theta_{s-1},a_{s-1}}\circ\cdots\circ L_{1,\theta_3,a_3}\circ L_{\theta_1,\theta_2,a_1}(x)} \\[1.25ex]
 & = & \int\!\!\cdots\!\!\int \mathbb{I}_{\{\sum_{i=1}^s \theta_i u_i\leq x \}}\;da_1(u_1)\cdots da_s(u_s)\\[1.25ex]
 & = & \int\!\!\cdots\!\!\int \mathbb{I}_{\{\sum_{i=1}^s \theta_i u_i\leq x \}}\;dF(u_1)\cdots d\BB(u_j)\cdots dF(u_s).
\end{eqnarray*}
Recalling that, according to (\ref{eq:sum-n}), $F_{\thetabar_{-j}}$ is  the CDF of $\sum_{i\neq j}\theta_iX_i$, integrating with respect to all variables except the $j$-th one, and bearing in mind that the condition of the indicator is equivalent to $\sum_{i\neq j} \theta_i u_i\leq x-\theta_ju_j,$ we can write
$$
\int\!\!\cdots\!\!\int \mathbb{I}_{\{\sum_{i=1}^s \theta_i u_i\leq x \}} \,\prod_{i\neq j}dF(u_i)=F_{\thetabar_{-j}}(x-\theta_ju_j).
$$
Summing over the set $C$ we obtain
$$
L_{\thetabar}'(F;\BB)=\sum_{j=1}^s \int F_{\thetabar_{-j}}(x-\theta_ju)\,d\BB(u).
$$
Now, note that
$$
F_{\thetabar_{-j}}(x-\theta_ju)=\int \mathbb{I}_{\{t\leq  x-\theta_ju\}}\,dF_{\thetabar_{-j}}(t)=\int \mathbb{I}_{\{u\leq  \tfrac{x-t}{\theta_j}\}}\,dF_{\thetabar_{-j}}(t),
$$
so that, by Fubini's Theorem,
\begin{eqnarray*}
\sum_{j=1}^s \int F_{\thetabar_{-j}}(x-\theta_ju)\,d\BB(u) & = &
   \sum_{j=1}^s \int \left(\int \mathbb{I}_{\{u\leq  \tfrac{x-t}{\theta_j}\}}\,dF_{\thetabar_{-j}}(t)\right)\,d\BB(u)\\
 & = & \sum_{j=1}^s  \int \left(\int \mathbb{I}_{\{u\leq  \tfrac{x-t}{\theta_j}\}}\,d\BB(u)\right)\,dF_{\thetabar_{-j}}(t)\\
 & = & \sum_{j=1}^s  \int \BB\left(\tfrac{x-t}{\theta_j}\right) \,dF_{\thetabar_{-j}}(t).
\end{eqnarray*}
The weak limit $\mathbb{H}_{\thetabar}^F$ is a centred Gaussian process, since it is obtained by linear transformations of the Gaussian process $\BB.$

The bilinearity of the covariance implies that 
\begin{eqnarray*}
\operatorname{Cov}\left(\mathbb{H}_{\thetabar}^F(x),\mathbb{H}_{\thetabar}^F(y)\right) & = & \sum_{j,k}\Omega_{j,k}(x,y) \\
 & = & \sum_{j,k}\iint\operatorname{Cov}\left(\BB\left(\tfrac{x-t}{\theta_j}\right),\,\BB\left(\tfrac{x-s}{\theta_k}\right)\right)\,dF_{\thetabar_{-j}}(t)dF_{\thetabar_{-k}}(s), \\
 & = & \sum_{j,k}\iint F\left(\tfrac{x-t}{\theta_j}\wedge\tfrac{x-s}{\theta_k}\right) - F\left(\tfrac{x-t}{\theta_j}\right)F\left(\tfrac{x-s}{\theta_k}\right)\,dF_{\thetabar_{-j}}(t)dF_{\thetabar_{-k}}(s).
\end{eqnarray*}
Using now the representation $F(z)=\int\mathbb{I}_{\{u\leq z\}}\,dF(u)$ for the integrands, and Fubini's Theorem, the conclusion follows.
\end{proof}

\label{ap:test-Cauchy}
\begin{proof}[Proof of Theorem~\ref{thm:Htilde}]
For a general CDF $F$, the process $\mathbb{H}_{\thetabar}^F-\mathbb{H}_{\etabar}^F$ is Gaussian, so it has trajectories that are almost surely continuous. Therefore, with probability 1 we have 
$$
\sup_{\{x:\,F_{\thetabar}(x)<F_{\etabar}(x)\}} \left(\mathbb{H}_{\thetabar}^F(x)-\mathbb{H}_{\etabar}^F(x)\right)=\sup_{\{x:\,F_{\thetabar}(x)<F_{\etabar}(x) \text{ a.e.}\}} \left(\mathbb{H}_{\thetabar}^F(x)-\mathbb{H}_{\etabar}^F(x)\right).
$$
So, without loss of generality, we may assume that $\widetilde{\mathcal H}_0=\{F:\,F_{\thetabar}(x)<F_{\etabar}(x),\,x\in\mathbb{R}\}$. As in the proof of Proposition~\ref{prop:Hadamard}, the operator $L_{\thetabar,\etabar}=(L_{\thetabar},L_{\etabar}):\ell^\infty(\R)\times \ell^\infty(\R)\to \ell^\infty(\R)\times \ell^\infty(\R)$ is Hadamard differentiable at the point $(F,F)$ on the direction of $(h_1,h_2)\in \ell^\infty(\R)\times \ell^\infty(\R)$, with derivative given by $L_{\thetabar,\etabar}((F,F);(h_1,h_2))=(L_{\thetabar}'(F;h_1),L_{\etabar}'(F;h_2))$. Then, the functional delta method gives 
$$
\sqrt n \left(\mathbb{F}_{n,\thetabar}-F_{\thetabar},\mathbb{F}_{n,\etabar}-F_{\etabar}\right)
\rightsquigarrow
(\mathbb{H}_{\thetabar}^F,\mathbb{H}_{\etabar}^F)
$$
in $\ell^\infty(\R)\times \ell^\infty(\R)$. Since the difference operator is linear and continuous, it follows that 
$$
\sqrt n\, \left((\mathbb{F}_{n,\thetabar} -\mathbb{F}_{n,\etabar})-(F_{\thetabar}-F_{\etabar})\right)\rightsquigarrow {\mathbb{H}_{\thetabar}^F-\mathbb{H}_{\etabar}^F}.
$$
Define the map $\mathcal M:\ell^\infty(\R)\to \ell^\infty(\R)$ by $\mathcal{M}(h)=h_+=\max(0,h)$. Clearly, taking into account that we are referring to CDFs, $\sup (\mathbb{F}_{n,\thetabar}-\mathbb{F}_{n,\etabar})=\sup \mathcal M(\mathbb{F}_{n,\thetabar}-\mathbb{F}_{n,\etabar})$. 
Now, when $F\in\mathcal{H}_0$, $\mathcal{M}(F_{\thetabar}-F_{\etabar})=0$, so we have that $\mathcal{M}(\mathbb{F}_{n,\thetabar}-\mathbb{F}_{n,\etabar})=\mathcal{M}(\mathbb{F}_{n,\thetabar}-\mathbb{F}_{n,\etabar})-\mathcal{M}(F_{\thetabar}-F_{\etabar})$. Following \citet[p. 384]{fang}, the map $\mathcal M$ is Hadamard directionally differentiable at $\delta=F_{\thetabar}-F_{\etabar}$, with derivative
$$
{\mathcal{M}'(\delta;h)(x)}=\left\{
\begin{array}{lcl}
h(x), & \quad & \text{ if $\delta(x)>0$ a.e.,}\\ [1ex]
\max(0,h(x)), & & \text{ if $\delta(x)=0$ a.e.,}\\[1ex]
0, & &\text{ if $\delta(x)<0$ a.e.}
\end{array}\right.
$$
Therefore, we may apply the functional delta method to the mapping $\mathcal{M}$, to obtain 
$$
\sqrt{n}\,\left(\mathcal{M}(\mathbb{F}_{n,\thetabar}-\mathbb{F}_{n,\etabar})-\mathcal{M}(F_{\thetabar}-F_{\etabar})\right)\rightsquigarrow 
{\mathcal{M}'(\delta;\mathbb{H}_{\thetabar}^F-\mathbb{H}_{\etabar}^F)},
$$ 
in $\ell^\infty(\mathbb{R})$.
Finally, $F\in\widetilde{\mathcal{H}}_0$ if and only if $\delta<0$, so $
{\mathcal{M}'(\delta;\mathbb{H}_{\thetabar}^F-\mathbb{H}_{\etabar}^F)}=0$. Accordingly, the continuous mapping theorem implies that $\sqrt{n}\,\sup\mathcal{M}(\mathbb{F}_{n,\thetabar}-\mathbb{F}_{n,\etabar})\to_p 0$. Differently, for $F=\mathcal C$, $\delta(x)=0$ for every $x\in\R$, hence, the Hadamard directional derivative coincides with $\mathcal{M}$, and so, we have that 
$\sqrt{n}\,\sup(\mathbb{F}_{n,\thetabar}-\mathbb{F}_{n,\etabar})\to_d \sup ({\mathbb{H}_{\thetabar}^{\mathcal{C}}-\mathbb{H}_{\etabar}^{\mathcal{C}}})$, which is a nonnegative random variable. 
\end{proof}

\begin{proof}[Proof of Proposition~\ref{prop:test-conv-Cauchy}]
{The first assertion is obvious.}
Note that, as $\mathbb{H}_{\thetabar}^{\mathcal{C}}-\mathbb{H}_{\etabar}^{\mathcal{C}}$ is a centred Gaussian process, $\sqrt{n}\,T_{\thetabar,\etabar}(\mathbb C_n)$ converges in distribution to a nonnegative random variable with strictly positive and finite quantile, namely, $c_{\alpha,n}\to c_\alpha<\infty$. If $F\in\widetilde{\mathcal{H}}_0$, $\sqrt{n}\,T_{\thetabar,\etabar}(\F_n)\to_p 0$, therefore the {second} assertion follows immediately. On the other hand, if $F\in\mathcal{H}_1$, ${T_{\thetabar,\etabar}}(\F_n)\to_{p}t_0>0$, therefore, $\sqrt{n}\,{T_{\thetabar,\etabar}}(\F_n)\to_{p}+\infty$, proving the {third} assertion.
\end{proof}

\begin{proof}[Proof of Proposition~\ref{prop:test-conv-boot}]
Assume $F\in\mathcal{H}_0$, then, as stated in (\ref{eq:upper1}), {$\sqrt{n}\,T_{\thetabar,\etabar}(\F_n)$ is bounded above by a sequence of random variables which converges in distribution to $\left\|\mathbb H_{\thetabar}^F-\mathbb H_{\etabar}^F\right\|_\infty$.}
{Since $\mathbb H_{\thetabar}^F-\mathbb H_{\etabar}^F$} is a centred Gaussian process, {then the} $(1-\alpha)$-quantile of {$\left\|\mathbb H_{\thetabar}^F-\mathbb H_{\etabar}^F\right\|_\infty$}, denoted as $c_\alpha$, is finite, positive and unique. 
{Therefore, under} $\mathcal{H}_1$, $T_{\thetabar,\etabar}(\F_n,)\to_{p}t_0>0$, {so}, $\sqrt{n}\,T_{\thetabar,\etabar}(\F_n)\to_{p}\infty$. 
The bootstrap process $\sqrt{n}\,(\F_n^b-\F_n)$ converges weakly almost surely to $\BB$, conditionally on the data. Then, the functional delta method for the bootstrap implies that, for each weight vector $\thetabar$, {$\sqrt{n}\,\left(T_{\thetabar}(\F_n^b)-T_{\thetabar}(\F_n)\right)$} converges weakly in probability to {$\mathbb{H}^F_{\thetabar}$}, conditionally on the data (see Theorem 2.9 in \cite{kosorok}). 
Therefore, we can conclude that $\mathcal{G}^b_n$ converges weakly in probability to {$\mathbb{H}_{\thetabar}^F-\mathbb{H}_{\etabar}^F$}, conditionally on the data. Finally, by the continuous mapping theorem, $\left\| \mathcal{G}^b_n\right\|_\infty\to_d {\left\|{\mathbb{H}_{\thetabar}^F-\mathbb{H}_{\etabar}^F}\right\|_\infty}$ conditionally on the data. 
The test rejects $\mathcal{H}_0$ if $\sqrt{n}\,T_{\thetabar,\etabar}(F_n)>c_{\alpha,n}$ given by (\ref{eq:critical}), and we now have that $c_{\alpha,n}\to_pc_\alpha$.
Therefore, under the null hypothesis $\mathcal{H}_0$, $\lim_{n\to\infty}P\left({T_{\thetabar,\etabar}}(\F_n)>c_{\alpha,n}\right)\leq \alpha$. Differently, if $F\in\mathcal{H}_1$, $\lim_{n\to\infty}P\left(T_{\thetabar,\etabar}(\F_n)>c_{\alpha,n}\right)=1$ because $c_\alpha$ is finite.
\end{proof} 
\end{appendix}


\bibliographystyle{abbrvnat} 
\bibliography{biblio}  
\section*{Funding}
P.E.O. acknowledges financial support by the Centre for Mathematics of the University of Coimbra (CMUC, https://doi.org/10.54499/UID/00324/2025) under the Portuguese Foundation for Science and Technology (FCT), Grants UID/00324/2025 and UID/PRR/00324/2025.
\section*{Competing interests}
The authors declare none.
\section*{Acknowledgements}The authors used ChatGPT for language editing of selected parts of the manuscript
and to assist with code checking for the numerical analyses.
\section*{Corresponding author contact}
Tommaso Lando, Department of Economics, University of Bergamo, via dei Caniana 2, 24127 Bergamo, Italy. Email: tommaso.lando@unibg.it.
\end{document}